\documentclass[runningheads]{llncs}

\usepackage{times}
\usepackage[latin1]{inputenc}
\usepackage[english]{babel}
\usepackage{amsmath}
\usepackage{colortbl}
\usepackage{txfonts}
\usepackage{xypic}
\usepackage{multirow}
\usepackage{textcomp}
\usepackage{calc}
\usepackage{stmaryrd}
\usepackage{marvosym}
\usepackage{pgf}
\usepackage[mathcal]{euscript}
\usepackage{caption}

\usepackage{ifthen}
\usepackage{mdwlist}
\usepackage{enumerate}
\usepackage{paralist}

\usepackage{tikz}
\usetikzlibrary{arrows,automata}
\usetikzlibrary{matrix}
\usetikzlibrary{decorations.pathmorphing}
\usetikzlibrary{decorations.markings}
\usetikzlibrary{decorations.shapes}

\DeclareMathAlphabet{\mathpzc}{OT1}{pzc}{m}{it}

\newcommand{\mrk}[1]{{#1}'}
\newcommand{\oldstack}[3]{%
{\ifthenelse{\equal{#1}{1}}{%
\mrk{#2}
}%
{#2}}_{#3}%
}
\newcommand{\stack}[3]{%
[%
{\ifthenelse{\equal{#1}{1}}{%
\mrk{#2}
}%
{#2}}\ {#3}%
]%
}

\newcommand{\va}[1]{\stackrel{#1}{\longrightarrow}}
\newcommand{\ourpath}[1]{\stackrel{#1}{\leadsto}}
\newcommand{\flush}[1]{\stackrel{#1}{\Longrightarrow}}
\newcommand{\lflush}[1]{\stackrel{#1}{\Longrightarrow}}

\newcommand{\cflush}[2]{\stackrel{{\color{#2}#1}}{\Longrightarrow}}

\newcommand{\chain}[3]{\langle^{#1} #2 {}^{#3} \rangle}

\newcommand{\config}[2]{\langle #1\ , \ #2 \rangle}
\newcommand{\triple}[3]{\langle #1\ , \ #2, \ #3 \rangle}
\newcommand{\symb}[1]{\mathop{symbol}(#1)}
\newcommand{\state}[1]{\mathop{state}(#1)}

\newcommand{\configtr}[2]{ #1 \downarrow #2}

\newcommand{\comp}[1]{ \stackrel {{#1}} \vdash }

\newcommand{\opl}{\rm OPL}
\newcommand{\opa}{\rm OPA}

\newcommand{\vpa}{\rm VPA}

\newcommand{\ofl}{{\rm$\omega$OPL}} 
\newcommand{\ofa}{{\rm$\omega$OPA}} 

\newcommand{\bfa}{{\rm$\omega$OPBA}} 
\newcommand{\bfae}{{\rm$\omega$OPBEA}} 
\newcommand{\mfa}{{\rm$\omega$OPMA}} 
\newcommand{\dbfa}{{\rm$\omega$DOPBA}} 
\newcommand{\dmfa}{{\rm$\omega$DOPMA}} 
\newcommand{\dbfae}{{\rm$\omega$DOPBEA}}
\newcommand{\bvpa}{{\rm$\omega$BVPA}}
\newcommand{\dbvpa}{{\rm$\omega$DBVPA}}

\newcommand{\lof}[1]{$\mathcal{L}$(#1)}

\begin{document}

\author{Federica Panella\inst{1}, Matteo Pradella\inst{1}, Violetta Lonati\inst{2}, Dino Mandrioli\inst{1}}
\institute{
  DEIB - Politecnico di Milano,
  via Ponzio 34/5, Milano, Italy\\
  \email{\{federica.panella, matteo.pradella, dino.mandrioli\}@polimi.it}
\and 
 DI - Universit\`a degli Studi di Milano,
  via Comelico 39/41, Milano, Italy\\
  \email{lonati@di.unimi.it}
}
\title{Operator Precedence $\omega$-languages}
\maketitle

\begin{abstract} 

$\omega$-languages are becoming more and more relevant nowadays when most applications are ``ever-running''. Recent literature, mainly under the motivation of widening the application of model checking techniques, extended the analysis of these languages from the simple regular ones to various classes of languages with ``visible syntax structure'', such as visibly pushdown languages (VPLs).
Operator precedence languages (OPLs), instead, were originally defined to support deterministic parsing and, though seemingly unrelated, exhibit interesting relations with these classes of languages: OPLs strictly include VPLs, enjoy all relevant closure properties and have been characterized by a suitable automata family and a logic notation.

In this paper we introduce operator precedence $\omega$-languages (\ofl s), investigating various acceptance criteria and their closure properties. Whereas some properties are natural extensions of those holding for regular languages, others required novel investigation techniques.
Application-oriented examples show the gain in expressiveness and verifiability offered by \ofl s w.r.t. smaller classes. 
\end{abstract}

\section{Introduction}
\label{sec:intro}

Languages of infinite strings, i.e. $\omega$-languages, have been
introduced to model nonterminating processes; thus they are becoming
more and more relevant nowadays when most applications are
``ever-running'', often in a distributed environment. The pioneering
work by B\"uchi and others investigated their main algebraic
properties in the context of finite state machines, pointing out
commonalities and differences w.r.t. the finite length counterpart
\cite{bib:Buchi1962,bib:Thomas1990a}. 

More recent literature, mainly under the motivation of widening the
application of model checking techniques to language classes as wide
as possible, extended this analysis to various classes of languages
with ``visible structure'', i.e., languages whose syntax structure is
immediately visible in their strings: parenthesis languages, tree
languages, visibly pushdown languages (VPLs) \cite{jacm/AlurM09} are
examples of such classes.

Operator precedence languages,
%, which we name after their inventor
%Floyd Languages (OPLs),
 instead, were defined by Floyd in the 1960s with the
original motivation of supporting deterministic parsing, which is
trivial for visible structure languages but is crucial for general
context-free languages such as programming
languages~\cite{Floyd1963}, where structure is often left implicit
(e.g. in arithmetic expressions). Recently, these seemingly unrelated
classes of languages have been shown to share most major features;
precisely OPLs strictly include VPLs and enjoy all the same closure
properties~\cite{Crespi-ReghizziM12}. This observation motivated
characterizing OPLs in terms of a suitable automata
family~\cite{LonatiMandrioliPradella2011a} and in terms of a logic
notation~\cite{LonatiMandrioliPradella2013a}, which was missing in
previous literature.

In this paper we further the investigation of OPLs properties to the
case of infinite strings, i.e., we introduce and study operator precedence
$\omega$-languages ({\ofl}s). As for other families, we consider
various acceptance criteria, their mutual expressiveness relations,
and their closure properties. Not surprisingly, some properties are
natural extensions of those holding for, say, regular languages or
VPLs, whereas others required different and novel investigation
techniques essentially due to the more general managing of the stack.
These closures and the decidability of the emptiness problem are a necessary step towards the
possibility of performing
infinite-state model checking.
Simple application-oriented examples show the considerable gain in
expressiveness and verifiability offered by {\ofl}s w.r.t. previous
classes.

The paper is organized as follows. The next section provides basic
concepts on operator precedence languages of finite-length words and on operator precedence
automata able to recognize
them. %as presented in~\cite{LonatiMandrioliPradella2011a}.
Section~\ref{sec:omegaLan} defines operator precedence automata which can deal with
infinite strings, analyzing various classical acceptance conditions
for $\omega$-abstract machines.  Section~\ref{sec:clProp} proves the
closure properties they enjoy w.r.t typical operations on
$\omega$-languages and shows also that the emptiness problem is
decidable for these formalisms.  Finally,
Section~\ref{sec:conclusions} draws some conclusions.

\section{Preliminaries}
\label{sec:prelim}

Operator precedence languages~\cite{Crespi-ReghizziM12,Floyd1963} have been characterized in terms of both a generative formalism (operator precedence grammars, OPGs) and an equivalent operational one (operator precedence automata, OPAs, named Floyd automata or FAs in ~\cite{LonatiMandrioliPradella2011a}), but in this paper we consider the latter, as it is better suited to model and verify nonterminating computations of systems. We first recall the basic notation and definition of operator precedence automata able to recognize words of finite length, as presented in~\cite{LonatiMandrioliPradella2011a}.

Let $\Sigma$ be an alphabet. The empty string is denoted $\varepsilon$.
Between the symbols of the alphabet three types of operator precedence (OP) binary relations can hold: \textit{yields} precedence, \textit{equal} in precedence and \textit{takes} precedence, denoted $\lessdot$, $\doteq$ and $\gtrdot$ respectively.
Notice that $\doteq$ is not necessarily an equivalence relation, and $\lessdot$ and $\gtrdot$ are not necessarily strict partial orders.
We use a special symbol \# not in $\Sigma$ to mark the beginning and the end of any string. This is consistent with the typical operator parsing technique that requires the lookback and lookahead of one character to determine the next action to perform~\cite{GruneJacobs:08}.
The initial \# can only yield precedence, and other symbols can only take precedence on the ending~\#.

\begin{definition}		
An \emph{operator precedence matrix} (OPM) $M$ over an alphabet $\Sigma$ is a  $|\Sigma \cup \{\#\}| \times |\Sigma\cup \{\#\}|$ array that with each ordered pair $(a, b)$ associates the set $M_{ab}$ of OP relations holding between $a$ and $b$. $M$ is \textit{conflict-free} iff $\forall a,b \in \Sigma, |M_{ab}|\leq 1$.
We call $(\Sigma, M)$ an operator precedence alphabet if $M$ is a conflict-free OPM on $\Sigma$.
\end{definition}

\noindent Between two OPMs $M_1$ and $M_2$, we define set inclusion and union:\vspace{-0.2cm}
%\begin{eqnarray}\label{OPMoperations}
\[
   \nonumber M_1\subseteq M_2 \text{ if } \forall a,b: (M_1)_{ab}\subseteq (M_2)_{ab}, \qquad
   \nonumber M=M_1\cup M_2   \text{ if }  \forall a,b: M_{ab}= (M_1)_{ab}\cup (M_2)_{ab}\vspace{-0.2cm}
\]
%   \nonumber M=M_1\cap M_2   \text{ if }  \forall a,b: M_{ab}= (M_1)_{ab}\cap (M_2)_{ab}
% \end{eqnarray}

If $M_{ab} = \{\circ\}$, with $\circ \in \{\lessdot, \doteq, \gtrdot\}$ ,we write $a \circ b$. For $u,v \in\Sigma^*$ we write $u \circ v$ if $u = xa$ and $v = by$ with $a \circ b$.
Two matrices are \emph{compatible} if their union is conflict-free.  A matrix is \textit{complete} if it contains no empty case.

In the following we assume that $M$ is \textit{$\dot=$-acyclic}, which means that $c_1 \doteq c_2 \doteq \dots \doteq c_k \doteq c_1$ does not hold for any $c_1, c_2, \dots, c_k \in \Sigma, k \geq 1$.

\begin{definition}\label{def:OPA}
A \emph{nondeterministic operator precedence automaton} (OPA) is a tuple
$\mathcal A = \langle \Sigma, M, Q, I, F, \delta \rangle $ where:
\begin{itemize}
\item $(\Sigma, M)$ is an operator precedence alphabet,
\item $Q$ is a set of states (disjoint from $\Sigma$),
\item $I \subseteq Q$ is a set of initial states,
\item $F \subseteq Q$ is a set of final states,
\item $\delta : Q \times ( \Sigma \cup Q) \rightarrow 2^Q$ is the transition function.
\end{itemize}
\end{definition}

\noindent The transition function can be seen as the union of two disjoint functions:\vspace{-0.1cm}
\[
\delta_{\text{push}}: Q \times \Sigma \rightarrow 2^Q
\qquad 
\delta_{\text{flush}}: Q \times Q \rightarrow 2^Q\vspace{-0.1cm}
\]
An OPA can be represented by a graph with $Q$ as the set of vertices and
$\Sigma \cup Q$ as the set of edge labels: 
there is an edge from state $q$ to state $p$ labeled by $a \in \Sigma$ if and only if $p \in \delta_{push}(q,a)$ and 
there is an edge from state $q$ to state $p$ labeled by $r \in Q$ if and only if $p \in \delta_{flush}(q,r)$.
To distinguish flush transitions from push transitions we denote the former ones by a double arrow.

To define the semantics of the automaton, we introduce some notation.
We use letters $p, q, p_i, q_i, \dots $ for states in $Q$ and 
we set $\mrk{\Sigma} = \{\mrk a \mid a \in \Sigma \}$; symbols in $\Sigma'$ are
called {\em marked} symbols.

Let $\Gamma$ be $ (\Sigma \cup \mrk \Sigma \cup \{\#\}) \times Q$; we denote symbols in $\Gamma$ as $\stack 0aq$,
$\stack 1aq$, or $\stack 0\#q$, respectively.
We set $\symb {\stack 0aq} = \symb {\stack 1aq} = a$, $\symb {\stack 0\#q}=\#$, and
$\state {\stack 0aq} = \state {\stack 1aq} = \state {\stack 0\#q} = q$.
Given a string $\beta = B_1 B_2 \dots B_n$ with $B_i \in \Gamma$, we set 
$\state \beta = \state{B_n}$.

A \emph{configuration} is any pair $C = \config {\beta} {w}$,
where $\beta = B_1 B_2 \dots B_n\in \Gamma^*$, $\symb{B_1} = \#$, and $w = a_1 a_2 \dots a_m \in \Sigma^*\#$.
A configuration represents both the contents $\beta$ of the stack and the part of input $w$ still to process.
%We also set $\tp C = \symb{B_n}$ and $\inpt C = a_1$.

A computation (run) of the automaton is a finite sequence of moves $C \vdash C_1$; there are three kinds of moves, depending on the precedence relation between $\symb {B_n}$ and~$a_1$:

\smallskip

\noindent {\bf push move:} if $\symb {B_n} \doteq a_1$ then 
$
C_1 = \config{\beta \stack 0{a_1}{q} }{a_2 \dots a_m}
$, with $q \in \delta_{push}(\state\beta,a_1)$;

\smallskip

\noindent {\bf mark move:} if $\symb {B_n} \lessdot \ a_1$ then
$
C_1 = \config{\beta \stack 1{a_1}{q} }{a_2 \dots a_m}
$, with $q \in \delta_{push}(\state\beta,a_1)$;

\smallskip

\noindent {\bf flush move:} if $\symb {B_n} \gtrdot a_1$ then
let $i$ the greatest index such that $\symb{B_i} \in \Sigma'$ (such index always exists). Then
$
C_1 = \config{B_1 B_2 \dots B_{i-2}\stack
0{\symb{B_{i-1}}}{q} }{a_1 a_2 \dots a_m}
$, with $q \in \delta_{flush}(\state{B_n},\state{B_{i-1}})$.

\smallskip

Push and mark moves both push the input symbol on the top of the stack, together with the new state computed by $\delta_{push}$; such moves differ only in the marking of the symbol on top of the stack.
The flush move is more complex: the symbols on the top of the stack are removed until the first marked symbol (\emph{included}),
and the state of the next symbol below them in the stack is updated by $\delta_{flush}$ according to the pair of states that delimit the portion of the stack to be removed; notice that in this move the input symbol is not consumed and it remains available for the following move.

Finally, we say that a configuration $\stack 0{\#}{q_I}$ is {\em starting} if $q_I \in I$ and
a  configuration $\stack 0{\#}{q_F}$ is {\em accepting} if $q_F \in F$.
The language accepted by the automaton is defined as:
\[
L(\mathcal A) = \left\{ x \mid  \config {\stack 0\#{q_I}} {x\#}  \comp * 
\config {\stack 0\#{q_F}} \# , q_I \in I, q_F \in F \right\}.
\]

\begin{remark}
The assumption on the $\doteq$-acyclicity has been introduced in previous literature~\cite{Crespi-ReghizziM12,LonatiMandrioliPradella2011a} to prevent the construction of operator precedence grammars with unbounded length of production's right hand sides (r.h.s.). Correspondingly, in presence of $\doteq$-cycles of an OPM, an OPA could be compelled to an unbounded growth of the stack before applying a flush move.
The $\doteq$-acyclicity hypothesis could be replaced by the weaker restriction of production's r.h.s. of bounded length in grammars and a bounded number of consecutive push moves in automata, or could be removed at all by allowing such unbounded forms of grammars -- e.g. with regular expressions as r.h.s.-- and automata.
In this paper we accept a minimal loss of generation\footnote{An example language that cannot be generated with an $\doteq$-acyclic OPM is the following: 
$\mathcal{L} = \{a^n {(b c)}^n \mid n\geq 0\} \cup \{b^n {(c a)}^n \mid n\geq 0\} \cup \{c^n (a b)^n \mid n\geq 0\} $} power and assume the simplifying assumption of  $\doteq$-acyclicity.
\end{remark}

 An OPA is \emph{deterministic} when $I$ is a singleton and $\delta_{\text{push}}(q,a)$ and $\delta_{\text{flush}}(q,p)$
have at most one element, for every $q,p \in Q$ and $a \in \Sigma$.

 An \emph{operator precedence transducer} can be defined in the usual way as a tuple $\mathcal T = \langle \Sigma, M, Q, I, F, O,\delta, \eta \rangle $ where $\Sigma, M, Q, I, F$ are defined as in Definition~\ref{def:OPA}, $O$ is a finite set of output symbols, the transition function $\delta$ and the output function $\eta$ are defined by $\langle \delta, \eta \rangle : Q \times (\Sigma \cup Q) \rightarrow \mathcal{P}_F(Q \times O^*) $, where $\mathcal{P}_F$ denotes the set of finite subsets of $ (Q \times O^*) $, and 
 $\langle \delta, \eta \rangle$ can be seen as the union of two  disjoint functions, 
  $\langle \delta_{\text{push}}, \eta_{\text{push}} \rangle : Q \times \Sigma \rightarrow \mathcal{P}_F(Q \times O^*) $ and $\langle \delta_{\text{flush}},\eta_{\text{flush}} \rangle : Q \times Q \rightarrow \mathcal{P}_F(Q \times O^*) $.
 
 A \textit{configuration} of the transducer is denoted $\configtr{\config{\beta}{w}}{z}$, where $C = \config {\beta} {w}$ is the configuration of the underlying OPA and the string after $\downarrow$ represents the output of the automaton in the configuration. The transition relation $ \vdash$  is naturally extended from OPAs, concatenating the output symbol produced at each move with those generated in the previous moves.
The \textit{transduction} $\tau : I^* \rightarrow \mathcal{P}_F(O^*)$ generated by $\mathcal T $ is defined by
\[\tau(x) = \left\{ z \mid \configtr{\config {\stack 0\#{q_I}}  {x\#}}{\varepsilon}  \comp * 
\configtr {\config {\stack 0\#{q_F}} \#}{z} , q_I \in I, q_F \in F \right\}
\]

\begin{example}\label{ex:db}
As an introductory example, consider a language of queries on a database expressed in relational algebra. We consider a subset of classical operators (union, intersection, selection $\sigma$, projection $\pi$ and natural join $\Join$). Just like mathematical operators, the relational operators have precedences between them: unary operators $\sigma$ and $\pi$ have highest priority, next highest is the $``multiplicative"$ operator $\Join$, lowest are the $``additive"$ operators $\cup$ and $\cap$.

Denote as $T$ the set of tables of the database and, for the sake of simplicity, let $E$ be a set of conditions for the unary operators. The OPA depicted in Figure~\ref{ex:db} accepts the language of queries without parentheses on the alphabet $\Sigma = T\cup$ $ \{\Join, \cup, \cap\} \cup\{\sigma, \pi\} \times E$, where we use letters $A, B, R\ldots$ for elements in $T$ and we write $\sigma_{\text{expr}}$ for a pair $(\sigma, \emph{expr})$ of selection with condition $\emph{expr}$ (similarly for projection $\pi_{\text{expr}}$). The same figure also shows an accepting computation on input $A \cup B \Join C \Join \pi_{\text{expr}} D$. 

Notice that the sentences of this language show the same structure as arithmetic expressions with prioritized operators and without parentheses, which cannot be represented by VPAs due to the particular shape of their OPM~\cite{Crespi-ReghizziM12}. % Thus, the language of this Example cannot be recognized by a VPA as well.
\end{example}

\vspace{-0.5cm}

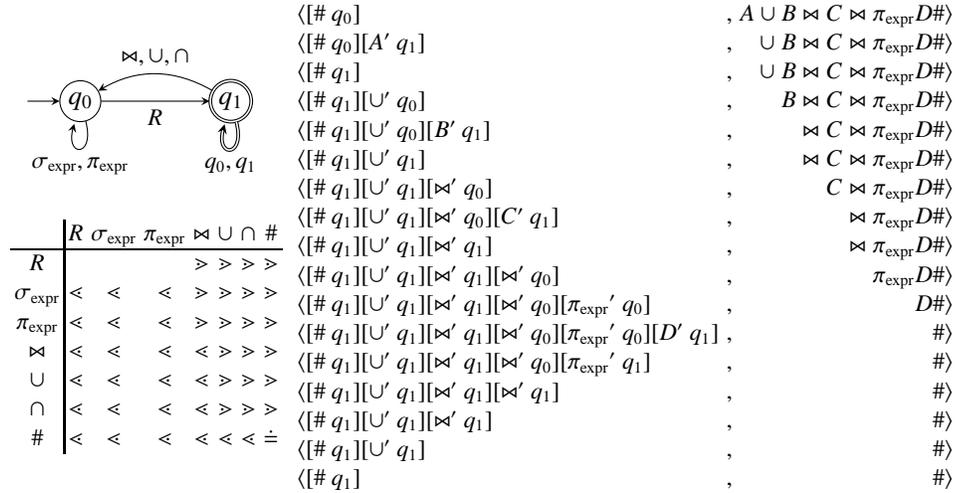
\begin{figure}[h!]
\begin{center}
\begin{tabular}{ll}
\begin{tabular}{l}
\begin{tikzpicture}[every edge/.style={draw,solid}, node distance=4cm, auto, 
                    every state/.style={draw=black!100,scale=0.5}, >=stealth]

\node[initial by arrow, initial text=,state] (S) {{\huge $q_0$}};
\node[state] (E) [right of=S, xshift=0cm, accepting] {{\huge $q_1$}};

\path[->]
(S) edge [loop below] node {$\sigma_{\text{expr}}, \pi_{\text{expr}}$} (S)
(S) edge [below] node {$R$} (E)
(E) edge [bend right, above]  node {$\Join, \cup, \cap$} (S)
(E) edge [loop below, double]  node {$q_0, q_1$} (E) ;
\end{tikzpicture}
\\
\\
$
\begin{array}{c|ccccccc}
      & R & \sigma_{\text{expr}}& \pi_{\text{expr}} & \Join & \cup&\cap & \# \\
\hline
R & & & &\gtrdot &\gtrdot &\gtrdot &\gtrdot \\
\sigma_{\text{expr}} & \lessdot & \lessdot & \lessdot &\gtrdot &\gtrdot &\gtrdot &\gtrdot  \\
\pi_{\text{expr}}   & \lessdot & \lessdot & \lessdot &\gtrdot &\gtrdot &\gtrdot &\gtrdot  \\
\Join  & \lessdot & \lessdot & \lessdot &\lessdot &\gtrdot &\gtrdot&\gtrdot  \\
\cup  & \lessdot & \lessdot & \lessdot &\lessdot &\gtrdot &\gtrdot&\gtrdot \\
\cap  & \lessdot & \lessdot & \lessdot &\lessdot &\gtrdot &\gtrdot&\gtrdot \\
\#       & \lessdot & \lessdot &\lessdot & \lessdot & \lessdot& \lessdot & \dot=  \\
\end{array}
$
\end{tabular}
& 
$
\begin{array}{lcr}
\langle \stack 0\# {q_0}     & , &A \cup B \Join C \Join \pi_{\text{expr}} D\# \rangle \\ 
\langle \stack 0\# {q_0} \stack 1 A {q_1}     & , & \cup \ B \Join C \Join \pi_{\text{expr}} D\# \rangle \\
\langle \stack 0\# {q_1}      & , &\cup \ B \Join C \Join \pi_{\text{expr}} D\# \rangle \\
\langle \stack 0\# {q_1} \stack 1 \cup {q_0}   & , &B \Join C \Join \pi_{\text{expr}} D\#\rangle \\
 \langle \stack 0\# {q_1} \stack 1 \cup {q_0} \stack{1}{B}{q_1}  & , &\Join C \Join \pi_{\text{expr}} D\# \rangle \\
 \langle \stack 0\# {q_1} \stack 1 \cup {q_1}      & , & \Join C \Join \pi_{\text{expr}} D\#\rangle \\
\langle \stack 0\# {q_1} \stack 1 \cup {q_1} \stack{1}{\Join}{q_0}     & ,&C \Join \pi_{\text{expr}} D\# \rangle \\
 \langle \stack 0\# {q_1} \stack 1 \cup {q_1} \stack{1}{\Join}{q_0} \stack 1 C {q_1}    & ,&\Join \pi_{\text{expr}} D\# \rangle \\
 \langle \stack 0\# {q_1} \stack 1 \cup {q_1} \stack{1}{\Join}{q_1}     & ,&\Join \pi_{\text{expr}} D\# \rangle \\
 \langle  \stack 0\# {q_1} \stack 1 \cup {q_1} \stack{1}{\Join}{q_1} \stack{1}{\Join}{q_0}     & ,&\pi_{\text{expr}} D\# \rangle \\
 \langle \stack 0\# {q_1} \stack 1 \cup {q_1} \stack{1}{\Join}{q_1}  \stack{1}{\Join}{q_0} \stack{1}{\pi_{\text{expr}}}{q_0}      & ,&D\# \rangle \\
 \langle  \stack 0\# {q_1} \stack 1 \cup {q_1} \stack{1}{\Join}{q_1} \stack{1}{\Join}{q_0} \stack{1}{\pi_{\text{expr}}}{q_0}\stack 1 D{q_1}   & ,&\# \rangle \\
 \langle \stack 0\# {q_1} \stack 1 \cup {q_1} \stack{1}{\Join}{q_1}  \stack{1}{\Join}{q_0}   \stack{1}{\pi_{\text{expr}}}{q_1}      & ,&\# \rangle \\
 \langle  \stack 0\# {q_1} \stack 1 \cup {q_1} \stack{1}{\Join}{q_1}  \stack{1}{\Join}{q_1}         & ,&\# \rangle \\
 \langle  \stack 0\# {q_1} \stack 1 \cup {q_1} \stack{1}{\Join}{q_1}           & ,&\# \rangle \\
 \langle \stack 0\# {q_1} \stack 1 \cup {q_1}  & ,&\# \rangle \\
 \langle \stack 0\# {q_1}   & ,&\# \rangle \\
\end{array}
$
\end{tabular}
\caption{Automaton, precedence matrix and example of computation for language of Example~\ref{ex:db}.}\label{fig:db}
\end{center}
\end{figure} 

\vspace{-0.7cm}
\noindent Let $(\Sigma, M)$ be a precedence alphabet.

\begin{definition}
A \emph{simple chain} is a word $a_0 a_1 a_2 \dots a_n a_{n+1}$,
written as
$
\chain {a_0} {a_1 a_2 \dots a_n} {a_{n+1}},
$
such that:
$a_0, a_{n+1} \in \Sigma \cup \{\#\}$, 
$a_i \in \Sigma$ for every $i: 1\leq i \leq n$, 
$M_{a_0 a_{n+1}} \neq \emptyset$,
and $a_0 \lessdot a_1 \doteq a_2 \dots a_{n-1} \doteq a_n \gtrdot a_{n+1}$.

A \emph{composed chain} is a word 
$a_0 x_0 a_1 x_1 a_2  \dots a_n x_n a_{n+1}$, 
where
$\chain {a_0}{a_1 a_2 \dots a_n}{a_{n+1}}$ is a simple chain, and
either $x_i = \varepsilon$ or $\chain {a_i} {x_i} {a_{i+1}}$ is a chain (simple or composed),
for every $i: 0\leq i \leq n$. 
Such a composed chain will be written as
$\chain {a_0} {x_0 a_1 x_1 a_2 \dots a_n x_n} {a_{n+1}}$.

A word $w$ over $(\Sigma, M)$ is \emph{compatible} with $M$ iff a)
for each pair of letters $c, d$, consecutive in $w$, $M_{cd} \neq \emptyset$, and b)
for each factor (substring) $x$ of  $\# w \#$
such that $x = a_0 x_0 a_1 x_1 a_2\dots$ $ a_n x_n a_{n+1}$
where
$a_0 \lessdot a_1 \doteq a_2 \dots a_{n-1} \doteq a_n \gtrdot a_{n+1}$ and,
for every $ 0\leq i \leq n$, either $x_i = \varepsilon$ or
 $\chain {a_i} {x_i} {a_{i+1}}$ is a chain (simple or composed), 
 $M_{a_0 a_{n+1}} \neq \emptyset$.
\end{definition}

\begin{definition}
Let $\mathcal A$ be an operator precedence automaton.
A \emph{support} for the simple chain
$\chain {a_0} {a_1 a_2 \dots a_n} {a_{n+1}}$
is any path in $\mathcal A$ of the form
\begin{equation}
\label{eq:simplechain}
%q_0
\va{a_0}{q_0}
\va{a_1}{q_1}
\va{}{}
\dots
\va{}q_{n-1}
\va{a_{n}}{q_n}
\flush{q_0} {q_{n+1}}
%\va{a_{n+1}}{}.
\end{equation}
Notice that the label of the last (and only) flush is exactly $q_0$, i.e. the first state of the path; this flush is executed because of relations $a_0 \lessdot a_1$ and $a_n
\gtrdot a_{n+1}$.

\noindent A \emph{support for the composed chain} 
$\chain {a_0} {x_0 a_1 x_1 a_2 \dots a_n x_n} {a_{n+1}}$
is any path in $\mathcal A$ of the form
\begin{equation}
\label{eq:compchain}
\va{a_0}{q_0}
\ourpath{x_0}{q'_0}
\va{a_1}{q_1}
\ourpath{x_1}{q'_1}
\va{a_2}{}
\dots
\va{a_n} {q_n}
\ourpath{x_n}{q'_n}
\flush{q'_0}{q_{n+1}}
%\va{a_{n+1}}{}
\end{equation}
where, for every $i: 0\leq i \leq n$: 
\begin{itemize}
\item if $x_i \neq \varepsilon$, then $\va{a_i}{q_i} \ourpath{x_i}{q'_i} $ 
is a support for the chain $\chain {a_i} {x_i} {a_{i+1}}$, i.e.,
it can be decomposed as $\va{a_i}{q_i}\ourpath{x_i}{q''_i} \flush{q_i}{q'_i}$.

\item if $x_i = \varepsilon$, then $q'_i = q_i$.
\end{itemize}

\noindent Notice that the label of the last flush is exactly $q'_0$.
\end{definition}

The chains fully determine the structure of the parsing of any
automaton on a word compatible with $M$, and hence the structure of the syntax tree of the word.
Indeed, if the automaton performs the computation
$ 
\config{\gamma \stack 0 a {q_0}} {x b y} \comp * 
\config{\gamma \stack 0 a {q}} {b y}
$ on a factor $a x b$ (with $\gamma \in \Gamma^*, y \in \Sigma^* \#$),
then $\chain axb$ 
is necessarily a chain over $(\Sigma, M)$ and there exists a support
like \eqref{eq:compchain} with $x = x_0 a_1 \dots a_n x_n$ and $q_{n+1} = q$.

\section{Operator precedence $\omega$-languages and automata}
\label{sec:omegaLan}

Let us now generalize operator precedence automata to deal with words of infinite length and to model nonterminating computations.

Traditionally, $\omega$-automata have been classified on the basis of the acceptance condition of infinite words they are equipped with. All acceptance conditions refer to the occurrence of states which are visited in a computation of the automaton, and they generally impose constraints on those states that are encountered infinitely (or also finitely) often during a run.
%The earliest definition for automata recognizing infinite-length words dates back to the work of B\"{u}chi~\cite{bib:Buchi1962}, while several other acceptance conditions (Muller~\cite{bib:Muller1963a}, Rabin~\cite{bib:Rabin1972a}, Streett~\cite{bib:Streett1982a}) have been subsequently introduced in order to express more and more powerful recognition modes.
Classical notions of acceptance (introduced by B\"{u}chi~\cite{bib:Buchi1962}, Muller~\cite{bib:Muller1963a}, Rabin~\cite{bib:Rabin1972a}, Streett~\cite{bib:Streett1982a}) can be naturally adapted to $\omega$-automata for operator precedence languages and can be characterized according to a peculiar acceptance component of the automaton on $\omega$-words. We first introduce the model of nondeterministic B\"{u}chi-operator precedence $\omega$-automata with acceptance by final state; other models are presented in Section~\ref{sec:otherModels}.
%More formally,

As usual, we denote by $\Sigma^\omega$ the set of infinite-length words over $\Sigma$. Thus, the symbol $\#$ occurs only at the beginning of an $\omega$-word.
Given a precedence alphabet $(\Sigma, M)$, the definition of an $\omega$-word compatible with the OPM $M$ and the notion of syntax tree of an infinite-length word are the natural extension of these concepts for finite strings.

%\begin{definition}
%An $\omega$-word is an infinite length word on the alphabet $\Sigma$ , and as for finite length words we use symbol \# to denote the beginning of the string. Differently from finite length words, there is no ending \# because the word is infinite.
%\end{definition}

\begin{definition}
\label{def:wFA-buchi}
A \emph{nondeterministic  B\"{u}chi-operator precedence $\omega$-automaton} (\bfa) is given by a tuple $\mathcal A = \langle \Sigma, M, Q, I, F, \delta \rangle $, where $\Sigma, Q, I, F, \delta$ are defined as for OPAs; the operator precedence matrix $M$ is restricted to be a $|\Sigma \cup \{ \# \}| \times |\Sigma|$ array, since $\omega$-words are not terminated by the delimiter \#.
% $F \subseteq Q$ is a set of final states.
\end{definition}

\emph{Configurations} and \emph{(infinite) runs} are defined as for operator precedence automata on finite-length words.
Then, let ``$\exists^\omega i$" be a shorthand for ``there exist infinitely many i" and let $\mathcal{S}$ be a run of the automaton on a given word $x \in \Sigma^\omega$. Define $
%\[
In(\mathcal S) = \{ q \in Q \mid \exists^\omega i  \ \config {\beta_i} {x_i} \in \mathcal{S}$ $\textit{with} \state{{\beta}_i}=q \}
%\]
$ 
as the set of states that occur infinitely often at the top of the stack of configurations in $\mathcal{S}$. 
A  run $\mathcal S$ of an \bfa \ on an infinite word $x \in \Sigma^\omega$ is \emph{successful} 
iff there exists a state $q_f \in F$ such that $q_f \in In(\mathcal S)$.
$\mathcal A$ \emph{accepts} $x \in \Sigma^\omega$ iff there is a successful run of $\mathcal{A}$ on x.
Furthermore, let the $\omega$-language \emph{recognized} by $\mathcal{A}$ be 
$
%\[
L(\mathcal A) = \left\{ x \in \Sigma^\omega \mid \text{$\mathcal{A}$ accepts x} \right\}
%\]
$.

Operator precedence \emph{$\omega$-transducers} are defined in the natural way as for finite-length words.

\subsection{Some examples}\label{sec:examples}

\begin{example}
\label{ex:BFA}
Consider a software system which is supposed to work forever and may serve interrupt requests issued by different users. 
The system can manage three types of interrupts with different levels of priority, that affect the order by which they are served by the system: pending lower priority interrupts are postponed in favor of higher priority ones.

This policy can be naturally specified by defining an alphabet of letters for ordinary procedures and for interrupt symbols, and by formalizing the priority level among the interrupt requests as OP relationships in the precedence matrix of an operator precedence automaton on infinite-length words: an interrupt yields precedence ($\lessdot$) to higher priority ones, which will be handled first, and takes precedence ($\gtrdot$) on lower priority requests, whose processing is then suspended.
Figure~\ref{fig:BFA} shows an \bfa\ with acceptance condition by final state which models the behavior of a system which may execute two functions denoted $a$ and $b$, that may be suspended by interrupts of types $int_0, int_1$ and $int_2$ with increasing level of priority. Calls and returns of the procedures are denoted $call_a, call_b, ret_a, ret_b$.
A request is actually served as soon as the corresponding interrupt symbol is flushed from the top of the stack.
Figure~\ref{fig:BFA} also presents the precedence matrix and an example computation of the system for the infinite string $call_a call_b ret_b call_b int_1 int_2 int_0 ret_b \ldots$

%\vspace{-0.3cm}
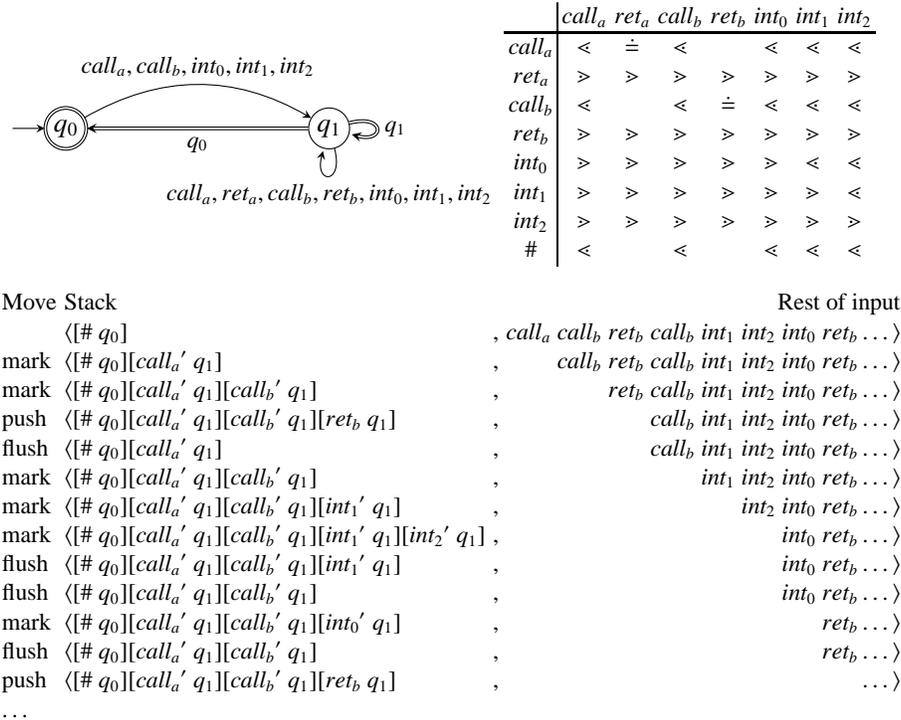
\begin{figure}[h!]
\begin{tabular}{m{.52\textwidth}m{.5\textwidth}} %{ll}
\small
\begin{tikzpicture}[every edge/.style={draw,solid}, node distance=4cm, auto, 
                    every state/.style={draw=black!100,scale=0.5}, >=stealth]

\node[initial by arrow, initial text=,state,accepting] (S) {{\huge $q_0$}};
\node[state] (E) [right of=S, xshift=3cm] {{\huge $q_1$}};

\path[->]
(S) edge [bend left]  node {$call_a, call_b, int_0, int_1, int_2$} (E)
(E) edge [loop right, double] node {$q_1$} (E)
(E) edge [loop below] node {$call_a, ret_a, call_b, ret_b, int_0, int_1, int_2$} (E)
(E) edge [below, double]  node {$q_0$} (S) ;
\end{tikzpicture}
&\quad
$
\begin{array}{c|ccccccc}
      & call_a & ret_a & call_b & ret_b & int_0 & int_1 & int_2   \\
\hline
call_a & \lessdot & \dot= & \lessdot &  & \lessdot & \lessdot & \lessdot \\
ret_a  & \gtrdot & \gtrdot & \gtrdot  & \gtrdot & \gtrdot & \gtrdot & \gtrdot  \\
call_b & \lessdot & & \lessdot & \dot=  & \lessdot & \lessdot & \lessdot  \\
ret_b  & \gtrdot & \gtrdot & \gtrdot  & \gtrdot & \gtrdot & \gtrdot & \gtrdot \\
int_0  & \gtrdot & \gtrdot & \gtrdot  & \gtrdot & \gtrdot & \lessdot & \lessdot  \\
int_1  & \gtrdot & \gtrdot & \gtrdot  & \gtrdot & \gtrdot & \gtrdot & \lessdot  \\
int_2  & \gtrdot & \gtrdot & \gtrdot  & \gtrdot & \gtrdot & \gtrdot & \gtrdot \\
\#       & \lessdot & & \lessdot & & \lessdot & \lessdot & \lessdot \\
\end{array}
$
\end{tabular}
\\
\begin{center}
$
\begin{array}{llcr}
\text{Move} & \text{Stack} & & \text{Rest of input}\\
 & \langle \stack 0\# {q_0}     & , &  
  call_a \ call_b \ ret_b \ call_b \ int_1 \  int_2 \ int_0 \ ret_b \dots \rangle \\ 
\text{mark} & \langle \stack 0\# {q_0} \stack 1 {call_a} {q_1}     & , &      call_b \ ret_b \ call_b \ int_1 \ int_2 \ int_0 \ ret_b\dots \rangle \\
\text{mark} & \langle \stack 0\# {q_0} \stack 1 {call_a} {q_1} \stack 1 {call_b} {q_1}     & , &     ret_b \ call_b \ int_1 \ int_2 \ int_0 \ ret_b \dots \rangle \\
\text{push} & \langle \stack 0\# {q_0} \stack 1 {call_a} {q_1} \stack 1 {call_b} {q_1} \stack 0 {ret_b} {q_1}     & , &       call_b \ int_1 \ int_2 \ int_0 \ ret_b \dots \rangle \\
\text{flush} & \langle \stack 0\# {q_0} \stack 1 {call_a} {q_1}      & , &      call_b \ int_1 \ int_2 \ int_0 \ ret_b  \dots \rangle \\
\text{mark} & \langle \stack 0\# {q_0} \stack 1 {call_a} {q_1} \stack 1 {call_b} {q_1}     & , &     \ int_1 \ int_2 \ int_0 \ ret_b \dots \rangle \\
\text{mark} & \langle \stack 0\# {q_0} \stack 1 {call_a} {q_1} \stack 1 {call_b} {q_1}  \stack{1}{int_1}{q_1}   & , &     int_2 \ int_0 \ ret_b  \dots \rangle \\
\text{mark} & \langle \stack 0\# {q_0} \stack 1 {call_a} {q_1} \stack 1 {call_b} {q_1}  \stack{1}{int_1}{q_1}  \stack{1}{int_2}{q_1}  & , &    int_0 \ ret_b \dots \rangle \\
\text{flush} & \langle \stack 0\# {q_0} \stack 1 {call_a} {q_1} \stack 1 {call_b} {q_1}  \stack{1}{int_1}{q_1}    & , &    int_0 \ ret_b \dots \rangle \\
\text{flush} & \langle \stack 0\# {q_0} \stack 1 {call_a} {q_1} \stack 1 {call_b} {q_1}   & , &    int_0 \ ret_b  \dots \rangle \\
\text{mark} & \langle \stack 0\# {q_0} \stack 1 {call_a} {q_1} \stack 1 {call_b} {q_1}  \stack{1}{int_0}{q_1}   & , &     ret_b \dots \rangle \\
\text{flush} & \langle \stack 0\# {q_0} \stack 1 {call_a} {q_1} \stack 1 {call_b} {q_1}     & , &   ret_b  \dots \rangle \\
\text{push} & \langle \stack 0\# {q_0} \stack 1 {call_a} {q_1} \stack 1 {call_b} {q_1}  \stack{0}{ret_b}{q_1}    & , &     \dots \rangle \\
\dots & & \\
\end{array}
$
\end{center}
\caption{Automaton, precedence matrix and example of computation for language of Example~\ref{ex:BFA}.}\label{fig:BFA}
\end{figure}
%\vspace{-0.5cm}
\end{example}

Several variations of the above policy can be specified as well by
similar \bfa s; e.g., we might wish to formalize that high
priority interrupts flush pending calls, whereas lower priority ones
let the system resume serving pending calls once the interrupt has
been served. We might also introduce an explicit symbol to formalize
the end of serving an interrupt and specify that some events are
disabled while serving interrupts with a given priority, etc.

\begin{example}\label{ex:version}
  Operator precedence automata on infinite-length words can also be
  used to model the run-time behavior of database systems, e.g., for
  modeling sequences of users' transactions with possible rollbacks.
  Other systems that exhibit an analogous behavior are revision
  control (or {\em versioning}) systems (such as subversion or git).
  As an example, consider a system for version management of files
  where a user can perform the following operations on documents: save
  them, access and modify them, undo one (or more) previous changes,
  restoring the previously saved version.

  The following alphabet represents the user's actions: $sv$ (for {\em
    save}), $wr$ (for {\em write}, i.e. the document is opened and
  modified), $ud$ (for a single {\em undo} operation), $rb$ (for a
  {\em rollback} operation, where all the changes occurred since the
  previously saved version are discarded.

An \bfa \ which models the traces of possible actions of the user on a given document is a single-state automaton $\langle \Sigma, M, \{q\}, \{q\}, \{q\}, \delta \rangle $, where $\Sigma = \{sv, rb, wr, ud\}$, $\delta_{\text{push}}(q, a) = q, \forall a \in \Sigma$ and $\delta_{\text{flush}}(q, q) = q$ and its OPM  is:
\[
M =
\begin{array}{c|cccc}
      & sv & rb & wr & ud   \\
\hline
sv & \lessdot & \dot= & \lessdot & \\
rb  & \gtrdot & \gtrdot & \gtrdot  & \gtrdot \\
wr & \lessdot & \gtrdot & \lessdot & \dot=    \\
ud  & \gtrdot & \gtrdot & \gtrdot  & \gtrdot \\
\#       & \lessdot & & \lessdot & \\
\end{array}
\]

Furthermore, one can even consider some specialized models of this
system, that represent various patterns of user behavior.  For
instance, one in which the user regularly backs her work up, so
that no more than $N$ changes which are not undone (denoted $wr$ as
before) can occur between any two consecutive checkpoints $sv$
(without any rollback $rb$ between
them). Figure~\ref{fig:wBFAversionRef} shows the corresponding \bfa \
with $N = 2$, with the same OPM $M$.

\begin{figure}[h!]
\centering
\begin{tikzpicture}[scale=1.2, transform shape,
                    every edge/.style={draw,solid}, node distance=4cm, auto, 
                    every state/.style={draw=black!100,scale=0.5}, >=stealth]

\node[initial by arrow, initial text=,state] (q0) {{\huge $q_0$}};
\node[state] (0) [right of=q0, xshift=0cm] {{\huge $0$}};
\node[state] (1) [right of=0, xshift=0cm] {{\huge $1$}};
\node[state] (2) [right of=1, xshift=0cm] {{\huge $2$}};
\node[state] (q1) [below of=0, xshift=0cm] {{\huge $q_1$}};
\node[state] (q2) [right of=q1, xshift=0cm] {{\huge $q_2$}};
\node[state] (q3) [right of=q2, xshift=0cm] {{\huge $q_3$}};
\node[state] (q4) [above of=1, xshift=0cm] {{\huge $q_4$}};

\path[->]
(q0) edge [ pos=0.5]  node[midway, fill=white]{$sv$} (0) 
(0) edge [ left]  node[midway, fill=white] {$wr$} (1) 
(0) edge [bend right, above]  node[ sloped, pos=0.3]{$rb$} (q1)
(0) edge [ bend left, above ]  node[sloped, midway]{$wr$} (q4) 
(0) edge [loop above] node {$sv$} (0)
(1) edge [bend right, above] node[midway, fill=white] {$wr$} (q4)
(1) edge [bend right, right] node[sloped,  fill=white] {$ud$} (q1)
(1) edge [bend right, right] node[sloped, fill=white] {$sv$} (0)
(1) edge [ left]  node[midway, fill=white] {$wr$} (2)
(1) edge [above, double]  node[pos=0.36, fill=white] {$0$} (q2)
(2) edge [bend right, above] node[sloped, pos=0.5] {$wr$} (q4)
(2) edge [ below, out=0, in=-90, controls=+(-50:3.5) and +(-50:3)] node[sloped, pos=0.6] {$ud$} (q1)
(2) edge [bend right, left] node[sloped, pos=0.53, fill=white] {$sv$} (0)
(2) edge [right, double]  node[midway] {$1$} (q3)
(q4) edge [loop left] node {$wr,ud$} (q4)
(q4) edge [loop above, double] node {$q_4$} (q4)
(q4) edge [ left, double]  node[pos=0.35] {$0$} (0)
(q4) edge [ left, double]  node[pos=0.35] {$1$} (1)
(q4) edge [ right, double ]  node[pos=0.35]{$2$} (2)
(q1) edge [ above, double]  node[midway, fill=white] {$0$} (0)
(q1) edge [ left, double]  node[midway] {$1$} (1)
(q1) edge [ below, double]  node[pos=0.3] {$2$} (2)
(q1) edge [ left, double]  node[midway] {$q_0$} (q0)
(q2) edge [ below]  node[midway] {$rb$} (q1)
(q3) edge [ below, double]  node[midway] {$0$} (q2);
\end{tikzpicture}\vspace{-0.5cm}
\caption{\bfa\ of Example~\ref{ex:version}, with $N=2$.}\label{fig:wBFAversionRef}
\end{figure}

\noindent States $0, 1$ and $2$ denote respectively the presence of zero, one and two unmatched changes between two symbols $sv$. All states of the \bfa\ final.

An example of computation on the string $sv\ wr\ ud\ rb\ sv\ wr\ wr\ ud\ sv\ wr\ rb\ wr\ sv\dots$ is shown in Figure~\ref{fig:wBFAversionRefComp}.
\begin{figure}[h!]
\begin{center}
$
\begin{array}{llcr}
\text{Move} & \text{Stack} & & \text{Rest of input}\\
 & \langle \stack 0\# {q_0}     & , &  
  sv\ wr\ ud\ rb\ sv\ wr\ wr\ ud\ sv\ wr\ rb\ wr\ sv \dots \rangle \\ 
\text{mark} & \langle \stack 0\# {q_0} \stack 1 {sv} {0}     & , &       wr\ ud\ rb\ sv\ wr\ wr\ ud\ sv\ wr\ rb\ wr\ sv \dots \rangle \\
\text{mark} & \langle \stack 0\# {q_0} \stack 1 {sv} {0} \stack 1 {wr} {1}     & , &      ud\ rb\ sv\ wr\ wr\ ud\ sv\ wr\ rb\ wr\ sv\dots\rangle \\
\text{push} & \langle \stack 0\# {q_0}  \stack 1 {sv} {0} \stack 1 {wr} {1}  \stack 0 {ud} {q_1}     & , &        rb\ sv\ wr\ wr\ ud\ sv\ wr\ rb\ wr\ sv\dots\rangle \\
\text{flush} & \langle \stack 0\# {q_0} \stack 1 {sv} {0}       & , &       rb\ sv\ wr\ wr\ ud\ sv\ wr\ rb\ wr\ sv\dots \rangle \\
\text{push} & \langle \stack 0\# {q_0} \stack 1 {sv} {0}  \stack 0 {rb} {q_1}     & , &    sv\ wr\ wr\ ud\ sv\ wr\ rb\ wr\ sv \dots \rangle \\
\text{flush} & \langle \stack 0\# {q_0}   & , &   sv\ wr\ wr\ ud\ sv\ wr\ rb\ wr\ sv\dots\rangle \\
\text{mark} & \langle \stack 0\# {q_0} \stack 1 {sv} {0}     & , &        wr\ wr\ ud\ sv\ wr\ rb\ wr\ sv \dots \rangle \\
\text{mark} & \langle \stack 0\# {q_0} \stack 1 {sv} {0} \stack 1 {wr} {1}     & , &       wr\ ud\ sv\ wr\ rb\ wr\ sv\dots\rangle \\
\text{mark} & \langle \stack 0\# {q_0} \stack 1 {sv} {0} \stack 1 {wr} {1}  \stack 1 {wr} {q_4}    & , &        ud\ sv\ wr\ rb\ wr\ sv\dots\rangle \\
\text{push} & \langle \stack 0\# {q_0} \stack 1 {sv} {0} \stack 1 {wr} {1}  \stack 1 {wr} {q_4}  \stack 0 {ud} {q_4}  & , &        sv\ wr\ rb\ wr\ sv\dots\rangle \\
\text{flush} & \langle \stack 0\# {q_0} \stack 1 {sv} {0}   \stack 1 {wr} {1}   & , &       sv\ wr\ rb\ wr\ sv \dots\rangle \\
\text{mark} & \langle \stack 0\# {q_0} \stack 1 {sv} {0}   \stack 1 {wr} {1} \stack 1 {sv} {0}    & , &        wr\ rb\ wr\ sv \dots\rangle \\
\text{mark} & \langle \stack 0\# {q_0} \stack 1 {sv} {0}   \stack 1 {wr} {1}  \stack 1 {sv} {0} \stack 1 {wr} {1}    & , &        rb\ wr\ sv \dots\rangle \\
\text{flush} & \langle \stack 0\# {q_0} \stack 1 {sv} {0}   \stack 1 {wr} {1}  \stack 1 {sv} {q_2}   & , &        rb\ wr\ sv \dots\rangle \\
\text{push} & \langle \stack 0\# {q_0} \stack 1 {sv} {0}   \stack 1 {wr} {1}  \stack 1 {sv} {q_2} \stack 0 {rb} {q_1}    & , &         wr\ sv\dots\rangle \\
\text{flush} & \langle \stack 0\# {q_0} \stack 1 {sv} {0}   \stack 1 {wr} {1}    & , &      wr\ sv\dots\rangle \\
\text{mark} & \langle \stack 0\# {q_0} \stack 1 {sv} {0}   \stack 1 {wr} {1}  \stack 1 {wr} {2}    & , &        sv\dots\rangle \\
\text{mark} & \langle \stack 0\# {q_0} \stack 1 {sv} {0}   \stack 1 {wr} {1}  \stack 1 {wr} {2} \stack 1 {sv} {0}    & , &      \dots\rangle \\
\dots & & \\
\end{array}
$
\end{center}
\caption{Example of computation for the specialized system of Example~\ref{ex:version}}\label{fig:wBFAversionRefComp}
\end{figure}
\end{example}

\subsection{Operator precedence $\omega$-languages and visibly pushdown $\omega$-languages}

Classical families of automata, like Visibly Pushdown
Automata~\cite{jacm/AlurM09}, imply several restrictions that hinder
them from being able to deal with the concept of precedence among
symbols. These restrictions make them unsuitable to define systems
like those of Section~\ref{sec:examples}, and in general all paradigms
based on a model of priorities.

 Noticeably, VPAs on infinite-length words are significantly
extended by the class of OPAs, since VPAs introduce a rigid
partitioning on the alphabet symbols which heavily constrains the
possible relationships among them: any letter cannot assume a role
dependent on the context (as an interrupt which can yield or take
precedence over another one depending on the mutual priority), and
this restriction has some consequences on their expressive power w.r.t
\ofl s.  Actually, as it happens for finite-word
languages~\cite{Crespi-ReghizziM12,LonatiMandrioliPradella2011a}, one
can prove the following result.

\begin{theorem}
The class of languages accepted by $\omega${\rm BVPA} (\emph{nondeterministic B\"uchi visibly pushdown $\omega$-automata}) 
is a proper subset of that accepted by \bfa.
\end{theorem}

The behavior of version management systems like those in Example~\ref{ex:version} too cannot be modeled by $\omega$VPAs since the shape of their matrix allows only one-to-one relationships between matching symbols (as do-undo actions on a single change, denoted $wr$ and $ud$), whereas the return to a previous version, undoing all the possible sequence of changes performed in the meanwhile, is represented by a many-to-one relationship (holding among symbols $wr$ and a single $rb$).

\subsection{Other automata models for operator precedence $\omega$-languages}
\label{sec:otherModels}

There are several possibilities to define other classes of $\omega$-languages. 
In order to do that we introduce the following general definition.

\begin{definition}
\label{def:wFA}
A \emph{nondeterministic operator precedence $\omega$-automaton} (\ofa) is given by a tuple $\mathcal A = \langle \Sigma, M, Q, I, \mathcal{F}, \delta \rangle $, where $\Sigma, Q, I, \delta$ are defined as for OPAs; the operator precedence matrix $M$ is restricted to be a $|\Sigma \cup \{ \# \}| \times |\Sigma|$ array, since $\omega$-words are not terminated by the delimiter \#; $\mathcal{F}$ is an acceptance component, distinctive of the class (B\"{u}chi, Muller,\ldots) the automaton belongs to.
\emph{Deterministic} \ofa\ are specified as for operator precedence automata on finite-length words.
\end{definition}

A run is {\em successful} if it satisfies an acceptance condition on $\mathcal{F}$ based on a specific recognizing mode.
\noindent $\mathcal A$ \emph{accepts} $x \in \Sigma^\omega$ iff there is a successful run of $\mathcal{A}$ on x.
Furthermore, let the $\omega$-language \emph{recognized} by $\mathcal{A}$ be 
$
%\[
L(\mathcal A) = \left\{ x \in \Sigma^\omega \mid \text{$\mathcal{A}$ accepts x} \right\}
%\]
$.

When $\mathcal F$ is a subset $F \subseteq Q$, Definition~\ref{def:wFA} leads to Definition~\ref{def:wFA-buchi} of B\"{u}chi-operator precedence $\omega$-automaton;
\bfae\ is a variant of \bfa\ obtained when using the following acceptance condition: a word is recognized if the automaton traverses final states with an \empty{empty stack} infinitely often. Formally, a run $\mathcal S$ of an \bfae\ is successful 
iff there exists a state $q_f \in F$ such that configurations with stack $\stack 0\#{q_f}$ occur infinitely often in $\mathcal S$.

\begin{proposition}
\label{prop:BFAvsBFAE}
 \lof{\bfae} $\subset$ \lof{\bfa}.
\end{proposition}

\begin{proof}
The inclusion is trivial by definition. 
To see why it is proper, one can consider for instance the language $L_{\text{repbdd}}$ (studied in \cite{jacm/AlurM09}) consisting of infinite words on the alphabet $\{a, \underline{a}\}$, which can be interpreted as a language of calls and returns of a procedure $a$, with the further constraint that there is always a finite number of pending calls. A nondeterministic \bfa\ with final state acceptance condition can nondeterministically guess which is the prefix of the word containing the last pending call, and then recognizes the language $(L_{\text{Dyck}}(a, \underline{a}))^\omega$ of correctly nested words.
An \bfae\ cannot recognize this language.
%, as we show in Section~\ref{sec:BVPAvsBFAE}.
In fact, it may accept a word iff it reaches infinitely often a final configuration with empty stack during the parsing. However, the automaton is never able to remove all the input symbols piled on the stack since it cannot flush the pending calls interspersed among the correctly nested letters $a$, otherwise it would either introduce conflicts in the OPM or it would not be able to verify that they are in finite number.

\end{proof}

The classical notion of acceptance for Muller automata can be likewise defined for~\ofa s.

\begin{definition} %[Muller acceptance condition]
A \emph{nondeterministic Muller-operator precedence automaton} (\mfa) is an \ofa\ $\langle \Sigma, M,$ $Q, I, \mathcal{F}, \delta \rangle $ whose acceptance component is a collection of subsets of $Q$, $\mathcal{F} = \mathcal{T}\subseteq 2^Q$, called the {\em table} of the automaton.

\noindent A run $\mathcal{S}$ of an \mfa\ on an infinite word $x \in \Sigma^\omega$ is $\emph{successful}$ iff $In(\mathcal S) \in \mathcal{T}$, i.e. the set of states occurring infinitely often on the stack is a set in the table $\mathcal{T}$.
\end{definition}

%We use the general term \emph{Floyd $\omega$-languages} (\ofl s) to denote languages recognized by \ofa s under any acceptance condition. 
%We denote the class of languages recognized by a given class $\mathcal{C}$ of \ofa s as $\mathcal{L}(\mathcal{C})$.

In the case of classical finite-state automata on infinite words, nondeterministic B\"{u}chi automata and nondeterministic Muller automata are equivalent and define the class of $\omega$-regular languages.
Traditionally, Muller automata have been introduced to provide an adequate acceptance mode for deterministic automata on $\omega$-words. In fact, 
deterministic B\"{u}chi automata cannot recognize all $\omega$-regular languages, whereas deterministic Muller automata are equivalent to nondeterministic B\"{u}chi ones~\cite{bib:Thomas1990a}.

For VPAs on infinite words, instead, the paper \cite{jacm/AlurM09} showed that the classical determinization algorithm of B\"{u}chi automata into deterministic Muller automata is no longer valid, and deterministic Muller $\omega$VPAs are strictly less powerful than nondeterministic B\"uchi $\omega$VPAs.
A similar relationship holds for \ofa s too. 
%More generally, the inclusion relationship between the various classes of \ofa\ are summarized in the following statement, whose proof can be found in~\cite[Chapter 4]{PanellaMasterThesis2011}.
%
%\begin{theorem}
% \lof{\dbfa} $\subset$ \lof{\dmfa} $\subset$  \lof{\bfa} $=$ \lof{\mfa}
%and all inclusions are proper.
%%\item \lof{\bfae} $ \nsupseteq$ \lof{\dmfa}.
%\end{theorem}

The relationships among languages recognized by the different classes of operator precedence $\omega$-automata 
%-- nondeterministic \bfa s ($\mathcal{L}$(\bfa)), deterministic \bfa s ($\mathcal{L}$(\dbfa)), nondeterministic \mfa s ($\mathcal{L}$(\mfa)), deterministic \mfa s ($\mathcal{L}$(\dmfa)), nondeterministic \bfa s with E acceptance condition ($\mathcal{L}$(\bfa E)) 
and visibly pushdown $\omega$-languages are summarized in the structure of Figure~\ref{fig:hierarchy}, where \dbfae, \dbfa\ and \dmfa\  denote the classes of deterministic \bfae s, deterministic \bfa s and deterministic \mfa s respectively. 
The detailed proofs of the strict containment relations holding among the classes \lof{\bfa}, \lof{\bfae}, \lof{\dbfa}, \lof{\dmfa} and \lof{\bvpa} in Figure~\ref{fig:hierarchy} are presented in~\cite[Chapter 4]{PanellaMasterThesis2011} and we do not report them here again for space reasons.
In the following sections we provide the proofs regarding the relationships between the strict containment relations among the other classes in Figure~\ref{fig:hierarchy} and the relationships between those classes which are not comparable
(i.e., those linked with dashed lines in the figure), which are not included in~\cite{PanellaMasterThesis2011}.

%\vspace{-0.3cm}
\tikzset{
	path/.style={dotted},
	every edge/.style={draw,solid},
	normal/.style={solid},
}

\begin{figure}[h!]
\begin{center}
\begin{tikzpicture}[scale=0.85, sibling distance=3.5cm ]
\node (root) {$\mathcal{L}$(\bfa) $\equiv$ $\mathcal{L}$(\mfa)}
child {
			node (OPBEA) {$\mathcal{L}$(\bfae)}
			[level distance=3cm]
				child {
						node (DOPBEA) {$\mathcal{L}$(\dbfae)}
				}
}
child {
			node (DOPMA) {$\mathcal{L}$(\dmfa)}
				child {
					node (DOPBA) {$\mathcal{L}$(\dbfa)}
				}
}
child {
			node (BVPA) {$\mathcal{L}$(\bvpa)}
			[level distance=3cm]
				child {
					node (DBVPA) {$\mathcal{L}$(\dbvpa)}
				}
}
;
\draw[dashed] (OPBEA) -- (DOPBA);
\draw[dashed] (BVPA) to [out=+90, in=+90] (OPBEA);
\draw[dashed] (BVPA) -- (DOPMA);
\draw[dashed] (BVPA) -- (DOPBA);
\draw[dashed]  (OPBEA) --   node[black, fill=white] {$\nsupseteq$}  (DOPMA);
\draw[solid] (DBVPA) -- (DOPBA);
\draw[dashed] (DBVPA) to [out=+175, in=-80, fill=white]  (OPBEA);
%\draw[dashed] (OPBEA) to [out=250, in=250, controls=+(+50:-7) and +(+40:-4)] (DBVPA);
\draw[dashed] (DBVPA) -- (DOPBEA);
\draw[dashed] (DOPBEA) to [out=+5, in=-100, fill=white]  (BVPA);
%\draw[dashed] (BVPA) to [out=-70, in=-70, controls=+(-50:7) and +(-40:4)] (DOPBEA);
\draw[solid] (DOPBEA) -- (DOPBA);

\end{tikzpicture}
\end{center} % \vspace{-0.5cm}
\caption{Containment relations for \ofl s. Solid lines denote strict
inclusions; dashed lines link classes which are not comparable. It is still
open whether $\mathcal{L}$(\bfae) $\subseteq$ $\mathcal{L}$({\dmfa}) or
not.}\label{fig:hierarchy}
\end{figure}
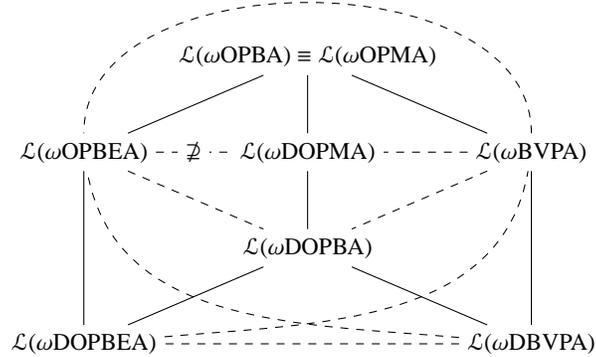

\subsection{Comparison between \lof{\bvpa} and \lof{\bfae}}
\label{sec:BVPAvsBFAE}

\lof{\bvpa} and \lof{\bfae} are not comparable.

\begin{itemize}
\item \lof{\bvpa} $\nsubseteq$ \lof{\bfae}\\
Consider the language $L_{\text{repbdd}}$ (studied in \cite{jacm/AlurM09}) consisting of infinite words on the alphabet $\{a, \underline{a}\}$, which can be interpreted as a language of calls and returns of a procedure $a$, with the further constraint that there is only a finite number of pending calls. 
An \bvpa \ can accept this language: it nondeterministically guesses which is the prefix of the string containing the last pending call, and it can subsequently recognize the language $(L_{\text{Dyck}}(a, \underline{a}))^\omega$ of correctly nested words.

An \bfae \ automaton cannot recognize this language, 
as seen in the proof of Proposition~\ref{prop:BFAvsBFAE}.

\item \lof{\bvpa} $\nsupseteq$ \lof{\bfae}\\
Consider the system introduced in Example 4 of~\cite{LonatiMandrioliPradella2011a} which describes the stack management of a programming language able to handle nested exceptions. 
No \bvpa  \ can express the language of the infinite computations of this system because of the shape of the precedence matrix, which is not compatible with the matrix of a VPA.

The automaton presented in the figure of this Example 4, which is able to recognize this language, instead, can be interpreted as an \bfae.
It is deterministic by construction, thus also \lof{\bvpa} $\nsupseteq$ \lof{\dbfae}.\\
Note also that the same automaton can be considered as an \bfa: since it is deterministic, there exists an \dbfa \ able to model this system, and \lof{\bvpa} $\nsupseteq$ \lof{\dbfa}.
Moreover, since \lof{\dbfa} $\subseteq$ \lof{\dmfa}, an automaton \dmfa \ can recognize it too; thus \lof{\bvpa} $\nsupseteq$ \lof{\dmfa}.

\end{itemize}

\subsection{Comparison between \lof{\bvpa} and \lof{\dmfa}}
\label{sec:BVPAvsDMFA}

\lof{\bvpa} and \lof{\dmfa} are not comparable.
\begin{itemize}
\item \lof{\bvpa} $\nsubseteq$ \lof{\dmfa}\\
No \dmfa \ can recognize the language $L_{\text{repbdd}}$ (the proof can be found in~\cite{PanellaMasterThesis2011}), whereas an \bvpa \ can accept it (see~\cite{jacm/AlurM09}).

\item \lof{\bvpa} $\nsupseteq$ \lof{\dmfa}\\
See Section~\ref{sec:BVPAvsBFAE}
\end{itemize}

\subsection{Comparison between \lof{\bvpa} and \lof{\dbfa}}
\label{sec:BVPAvsDBFA}

\lof{\bvpa} and \lof{\dbfa} are not comparable.

\begin{itemize}
\item \lof{\bvpa} $\nsubseteq$ \lof{\dbfa}\\
Consider the language on the alphabet $\Sigma = \{a,b\}$:
\begin{equation}
\label{eq:Lfinitely_a}
L_1 = \{  \alpha \in \Sigma^\omega : \alpha \text{ contains finitely many letters a } \}
\end{equation}
It can be recognized by an \bvpa, but no \dbfa \ can accept it.

In fact, an \bvpa \ can recognize words of $L_1$ finding nondeterministically the last letter $a$ in a word and then reading suffix $b^{\omega}$.

The proof that no \dbfa \ can recognize $L_1$ resembles the classical proof (see e.g.~\cite{bib:Thomas1990a}) that deterministic B\"{u}chi finite-state automata are strictly weaker than nondeterministic B\"{u}chi finite-state ones. We outline here the proof for the sake of completeness.\\
Assume that there exists an \dbfa \ $\mathcal B$ which recognizes $L_1$.

Notice that, in general, according to the definition of push/mark/flush moves of an operator precedence automaton (finite or $\omega$), given any configuration $C = \config {\beta} {w}$, the state piled up at the top of the stack with a transition $\config {\beta} {w} \vdash \config {\beta'} {w'}$, namely $\state{\beta'}$, is exactly the state reached by the automaton on its state-graph. Thus, during a run on a word $x \in \Sigma^\omega$, configurations with stack $\beta_i$ with $\state{{\beta}_i} \in F$ occur infinitely often iff the automaton visits infinitely often states in $F$ in its graph.

Now, the infinite word $x = b^\omega$ belongs to $L_1$, since it contains no (and then a finite number of) letters $a$. Then, there exists a unique run of $\mathcal B$ on this string which visits infinitely often final states. Let $b^{n_1}$ be the prefix read by $\mathcal B$ until the first visited final state.\\
But also $b^{n_1}ab^\omega$ belongs to $L_1$, hence there exists a final state reached reading the prefix $b^{n_1}a b^{n_2}$, for some $n_2 \in N$.\\
In general, one can find a sequence of  finite words $b^{n_1}a b^{n_2}\dots ab^{n_k}, (k\geq 1)$ such that the automaton has a unique run on them, and for each such runs it reaches a final state (placing it at the top of the stack) after reading every prefix $b^{n_1}a b^{n_2}\dots ab^{n_i}, \forall \text{ } i\leq k$. Therefore, there exists a (unique) run of $\mathcal A$ on the $\omega$-word $w = b^{n_1}ab^{n_2}\dots$ such that $\mathcal A$ visits infinitely often final states, and thus reaches infinitely often configurations $C = \config {\beta} {w}$ with $\state{{\beta}} \in F$.\\
However, $w$ cannot be accepted by $\mathcal B$ since it contains infinitely many letters $a$, and this is a contradiction.

\item \lof{\bvpa} $\nsupseteq$ \lof{\dbfa}\\
See Section~\ref{sec:BVPAvsBFAE}
\end{itemize}

\subsection{Comparison between \lof{\bvpa} and \lof{\dbfae}}
\label{sec:BVPAvsDBFAE}

\lof{\bvpa} and \lof{\dbfae} are not comparable.
\begin{itemize}
\item \lof{\bvpa} $\nsubseteq$ \lof{\dbfae}\\
If \lof{\bvpa} $\subseteq$ \lof{\dbfae}, then \lof{\bvpa} $\subseteq$ \lof{\bfae} since \lof{\dbfae} is a subclass of \lof{\bfae}. This, however, contradicts the fact that \lof{\bvpa} and \lof{\bfae} are not comparable.
\item \lof{\bvpa} $\nsupseteq$ \lof{\dbfae}\\
See Section~\ref{sec:BVPAvsBFAE}
\end{itemize}

\subsection{Comparison between \lof{\bfae} and \lof{\dbfa}}
\label{sec:BFAEvsDBFA}

\lof{\bfae} and \lof{\dbfa} are not comparable.
\begin{itemize}
\item \lof{\bfae} $\nsubseteq$ \lof{\dbfa}\\
Language $L_1 $ (Equation~\ref{eq:Lfinitely_a}) cannot be recognized by an \dbfa \ (see Section~\ref{sec:BVPAvsDBFA}), but there exists an \bfae\  accepting it, depicted in Figure~\ref{fig:BFAEfinitely_a} along with its precedence matrix (where $\circ \in \{\lessdot, \doteq, \gtrdot\}$ can be any precedence relation):

\begin{figure}[h!]
\begin{center}
\begin{tabular}{m{.2\textwidth}m{.2\textwidth}}
$
\begin{array}{c|cc}
      & a & b   \\
\hline
a & \circ & \gtrdot \\
b & \circ & \gtrdot \\
\#       & \lessdot & \lessdot \\
\end{array}
$
&
\begin{tikzpicture}[every edge/.style={draw,solid}, node distance=4cm, auto, 
                    every state/.style={draw=black!100,scale=0.5}, >=stealth]

\node[initial by arrow, initial text=,state] (q0) {{\huge $q_0$}};
\node[state] (q1) [accepting, right of=q0, xshift=0cm] {{\huge $q_1$}};

\path[->]
(q0) edge [above]  node {$b$} (q1) 
(q0) edge [loop above]  node {$a,b$} (q0) 
(q0) edge [loop below, double]  node {$q_0$} (q0) 
(q1) edge [loop above]  node {$b$} (q1);
\end{tikzpicture}
\end{tabular}
\end{center}
\caption{\bfae \ recognizing $L_1 = \{  \alpha \in \Sigma^\omega: \alpha \text{ contains finitely many letters $a$} \}$ and its OPM.}
\label{fig:BFAEfinitely_a}
\end{figure}
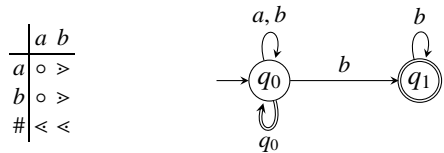

\item \lof{\bfae} $\nsupseteq$ \lof{\dbfa}\\
Let $L_2$ be the language $a^2 {L_3}^{\omega}$ with $L_3 = \{ a^k b^k \mid k\geq 1\}$ and where, in general, for a set of finite words $L \subseteq A^*$, one defines $ L^\omega = \{ \alpha \in A^\omega \mid \alpha = w_0 w_1\dots \text{ with } w_i \in  L \text{ for } i\geq 0 \}$.

No \bfae \  can recognize this language. Indeed, words in $L_3$ can be recognized only with the OPM $M$ depicted in Figure~\ref{fig:OPMa2anbnDBFA}, where $\circ \in \{\lessdot, \doteq, \gtrdot\}$ can be any precedence relation: clearly, using any other OPM there exist words in $L_3$ and $L_2 = a^2 {L_3}^{\omega}$ which could not be recognized. Thus, because of the OP relation $a \lessdot a$, an \bfae \ piles up on the stack the first sequence $a^2$ of a word and cannot remove it afterwards; hence it cannot empty the stack infinitely often to accept a string in $L_2$.

\begin{figure} [h!]
\begin{center}
$
\begin{array}{c|cc}
      & a & b  \\
\hline
a & \lessdot & \dot= \\
b  & \circ & \gtrdot \\
\#       & \lessdot & \\
\end{array}
$
\end{center}\caption{OPM for language $L_2$ of Section~\ref{sec:BFAEvsDBFA}.}\label{fig:OPMa2anbnDBFA}
\end{figure}

There is, however, an \dbfa \  that recognizes such a language (Figure~\ref{fig:a2anbnDBFA}). Incidentally notice that, since \lof{\dbfa} $\subseteq$ \lof{\dmfa}, an automaton \dmfa  \ can recognize it too; thus \lof{\bfae} $ \nsupseteq$ \lof{\dmfa}.

\begin{figure}[h!]
\centering
\begin{tikzpicture}[every edge/.style={draw,solid}, node distance=4cm, auto, 
                    every state/.style={draw=black!100,scale=0.5}, >=stealth]

\node[initial by arrow, initial text=,state] (q0) {{\huge $q_0$}};
\node[state] (q1) [right of=q0, xshift=0cm] {{\huge $q_1$}};
\node[state,accepting] (q2) [right of=q1, xshift=0cm] {{\huge $q_2$}};
\node[state] (q3) [right of=q2, xshift=0cm] {{\huge $q_3$}};

\path[->]
(q0) edge [above]  node {$a$} (q1) 
(q1) edge [above]  node {$a$} (q2) 
(q2) edge [above]  node {$a$} (q3) 
(q3) edge [loop below, double] node {$q_3$} (q3)
(q3) edge [loop right] node {$a,b$} (q3)
(q3) edge [bend left, double]  node {$q_2$} (q2) ;
\end{tikzpicture}
\caption{\dbfa \ recognizing language $L_2$ of Section~\ref{sec:BFAEvsDBFA}.}\label{fig:a2anbnDBFA}
\end{figure}
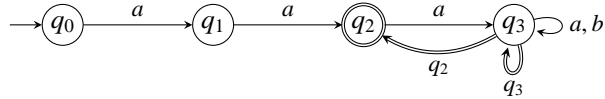
\end{itemize}

\subsection{Comparison between \lof{\bfae} and \lof{\dbvpa}}
\label{sec:BFAEvsDBVPA}

\lof{\bfae} and \lof{\dbvpa} are not comparable.
\begin{itemize}
\item \lof{\bfae} $\nsubseteq$ \lof{\dbvpa}\\
If \lof{\bfae} $\subseteq$ \lof{\dbvpa}, then \lof{\bfae} $\subseteq$ \lof{\bvpa} since \lof{\dbvpa} is a subclass of \lof{\bvpa}. This, however, contradicts the fact that \lof{\bfae} and \lof{\bvpa} are not comparable.
\item \lof{\bfae} $\nsupseteq$ \lof{\dbvpa}\\
Let $L =\Sigma^{\omega}$ with $\Sigma = \{a,b\}$ where the precedence relations between the symbols of the alphabet are represented by the OPM $M$ in Figure~\ref{fig:OPMBFAEvsDBVPA}, i.e. $\Sigma$ coincides with the call alphabet $\Sigma_c$ of a \vpa. $L$ can be recognized by an \dbvpa\ that has both input letters $a$ and $b$ as call symbols, but it cannot be recognized by any (nondeterministic or deterministic) \bfae\ with OPM $M$. Thus \lof{\bfae} $\nsupseteq$ \lof{\dbvpa} and \lof{\dbfae} $\nsupseteq$ \lof{\dbvpa}.

\begin{figure} [h!]
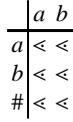

\begin{center}
$
\begin{array}{c|cc}
      & a & b  \\
\hline
a & \lessdot & \lessdot \\
b  & \lessdot & \lessdot \\
\#  &\lessdot   & \lessdot  \\
\end{array}
$
\end{center}\caption{OPM for language $L$ of Section~\ref{sec:BFAEvsDBVPA}.}\label{fig:OPMBFAEvsDBVPA}
\end{figure}

\end{itemize}

\subsection{Comparison between \lof{\bfae} and \lof{\dbfae}}
\label{sec:BFAEvsDBFAE}

\lof{\dbfae} $\subset$ \lof{\bfae}\\
The inclusion between the two classes is strict. Consider, in fact, language $L_1$ of Equation~\ref{eq:Lfinitely_a}: $L_1$ can be recognized by an \bfae, but no \dbfae\ can recognize it (the proof is analogous to that presented for \dbfa s in Section~\ref{sec:BVPAvsDBFA}).

% %
\subsection{Comparison between \lof{\dbfae} and \lof{\dbfa}}
\label{sec:DBFAEvsDBFA}

\lof{\dbfae} $\subset$ \lof{\dbfa}\\
The inclusion holds since for any \dbfae \ there exists an \dbfa\ which recognizes the same language: the \dbfa\ simply keeps in the states information on the evolution of the stack marking those states which are reached with empty stack in the \dbfae \ (in particular, the proof that \lof{\bfae} $\subseteq$ \lof{\bfa} in~\cite{PanellaMasterThesis2011} describes how to define an \bfa\ $\mathcal{\tilde{A}}$ equivalent to a given \bfae \ $\mathcal{A}$, and $\mathcal{\tilde{A}}$ is deterministic if $\mathcal{A}$ is deterministic).

The inclusion is strict: language $L_2$ in Section~\ref{sec:BFAEvsDBFA}, for instance, belongs to \lof{\dbfa} but it cannot be recognized by any \dbfae.
\subsection{Comparison between \lof{\dbfae} and \lof{\dbvpa}}
\label{sec:DBFAEvsDBVPA}

\lof{\dbfae} and \lof{\dbvpa} are not comparable.
\begin{itemize}
\item \lof{\dbfae} $\nsubseteq$ \lof{\dbvpa}\\
If \lof{\dbfae} $\subseteq$ \lof{\dbvpa}, then \lof{\dbfae} $\subseteq$ \lof{\bvpa} since \lof{\dbvpa} is a subclass of \lof{\bvpa}. This, however, contradicts the fact that \lof{\dbfae} and \lof{\bvpa} are not comparable.
\item \lof{\dbfae} $\nsupseteq$ \lof{\dbvpa}\\
See Section~\ref{sec:BFAEvsDBVPA}.

\end{itemize}

\subsection{Comparison between \lof{\bvpa} and \lof{\dbvpa}}
\label{sec:BVPAvsDBVPA}

\lof{\dbvpa} $\subset$ \lof{\bvpa} \\
The inclusion is strict: no \dbvpa\ can recognize language $L_1$ of Equation~\ref{eq:Lfinitely_a}, whereas an \bvpa\ can accept it.

\subsection{Comparison between \lof{\dbfa} and \lof{\dbvpa}}
\label{sec:DBFAvsDBVPA}

\lof{\dbvpa} $\subset$ \lof{\dbfa}\\
Between \lof{\dbvpa} and \lof{\dbfa} the same relationship holds as for their corresponding nondeterministic counterparts; in particular the inclusion is strict, as for \bvpa s and \bfa s, as Section~\ref{sec:BVPAvsBFAE} presented a system that can be modeled by an \dbfa \ and by no \bvpa.

\section{Closure properties and emptiness problem}
\label{sec:clProp}

%Operator precedence $\omega$-automata generalize VPAs as they can be applied in various contexts which cannot be dealt with by VPAs, from processing of streaming structured or semistructured web documents where they may support a more flexible management of tags, up to the specification of operating systems or web or versioning applications which deal with priority-ordered tasks.

\lof{\bfa} enjoys all closure and decidability properties necessary to perform model checking; thus thanks to their greater expressive power, we believe that they represent a truly promising formalism for infinite-state model-checking.

In the first part of this section we focus on the most interesting closure properties of \ofa s, which are summarized in Table~\ref{tab:closureProp}, where they are compared with the properties enjoyed by VPAs on infinite-length words. Binary operations are considered between languages with compatible OPMs.

\vspace{-0.3cm}
\begin{table}
\centering
\begin{tabular}{|c|c|c|c|c|c|c|}
\hline
& $\mathcal{L}$(\dbfae)  & $\mathcal{L}$(\bfae) & $\mathcal{L}$(\dbfa)  & $\mathcal{L}$(\dmfa) &$\mathcal{L}$(\bfa) & $\mathcal{L}$(\bvpa)\\
%& $\mathcal{L}$(\dbfa) & $\mathcal{L}$(\dmfa) & $\mathcal{L}$(\bfa)$\equiv$$\mathcal{L}$(\mfa) & $\mathcal{L}$($\omega$BVPA)~\cite{jacm/AlurM09} \\
 \hline 
Intersection &Yes & Yes  &Yes & Yes & Yes & Yes\\
 \hline
Union &Yes & Yes  &Yes & Yes & Yes & Yes\\
\hline 
Complement & No & No & No & Yes &Yes  & Yes\\
\hline 
$L_1 \cdot L_2$ & No & No & No & No &Yes  & Yes\\
\hline 
\end{tabular}\medskip
\caption{Closure properties of families of $\omega$-languages. ($L_1 \cdot L_2$ denotes the concatenation of a language of finite-length words $L_1$ and an $\omega$-language $L_2$).\label{tab:closureProp}}
\end{table}
\vspace{-0.7cm}

Closure properties for \dbfa s (under complement and concatenation with an OPL) and \dmfa s are not discussed here because of space reasons, but they resemble proofs for classical families of $\omega$-automata and can anyhow be found in~\cite{PanellaMasterThesis2011}.
Closure properties for \dbfa s under intersection and union are presented in Section~\ref{sec:wDBFA}; closure properties for \bfae s and \dbfae s are presented in Section~\ref{sec:wBFAE} and Section~\ref{sec:wDBFAE}.

We consider in detail the main family \bfa. This class is closed under Boolean operations between languages with compatible precedence matrices and under concatenation with a language of finite words accepted by an OPA.
The emptiness problem is decidable for \ofa s in polynomial time because they can be interpreted as pushdown automata on infinite-length words:  e.g.~\cite{BurkartSteffen1992a} shows an algorithm that decides the alternation-free modal $\mu$-calculus for context-free processes, with linear complexity in the size of the system's representation; thus the emptiness problem for the intersection of the language recognized by a pushdown process and the language of a given property in this logic is decidable.
Closures under intersection and union hold for \bfa s as for classical $\omega$-regular languages and can be proved in a similar way~\cite{PanellaMasterThesis2011}. Closures under complementation and concatenation required novel investigation techniques. 
%To prove them we first need some preliminary notions and remarks, where simple chains, composed chains and their supports are defined as in~\cite{LonatiMandrioliPradella2011a}.

\subsection*{Closure under concatenation}
\label{sec:concatBFAF}

For classical families of automata (on finite or infinite-length words) the closure of the class of languages they recognize with respect to the operation of concatenation is traditionally proved resorting to a Thompson-like construction: given two automata that recognize languages of a given class, an automaton which accepts the concatenation of these languages is generally defined so that it may simulate the moves of the first automaton while reading the first word of the concatenation and, once it reaches some final state, it switches to the initial states of the second automaton to begin the recognition of words of the second language.

 This construction, however, is not adequate for the concatenation of a language of finite words recognized by a classical OPA and an \ofl\ (recognized by an \bfa).
In fact, a classical OPA accepts a finite word by reaching a final state and by emptying its stack thanks to the ending delimiter $\#$.
As regards the concatenation of a language recognized by an OPA and an $\omega$-language (accepted by an \bfa) whose words are not ended by $\#$, this condition is not necessarily guaranteed and it might be not possible to complete the recognition of a word of the first language simulating the behavior of its OPA according to the acceptance condition by final state and empty stack. As an example, for a language $L_1 \subseteq \Sigma^*$ and an $\omega$-language $L_2 = \{a^\omega\}$ with compatible precedence matrices such that all letters of the alphabet yield precedence to symbol $a$ (i.e. $b \lessdot a, \forall b \in \Sigma$), the symbols still on the stack after reading words in $L_1$ cannot be removed with flush moves before or during the parsing of the second word in the concatenation, since the precedence relation $\lessdot$ implies that the letters read are only pushed on the stack. Thus, the stack cannot be emptied after the reading of the first word, and this prevents to check if it actually belongs to the first language of the concatenation.

After reading the first finite word in the concatenation, it is not even possible to determine whether this word is accepted by checking if in its OPA there exists an ongoing run on it that could lead to a final state by flush moves induced by a potential delimiter~$\#$, since this control would require to know the states already reached and piled on the stack, which are not visible without emptying the stack itself.

Closure under concatenation for the class of languages accepted by \bfa s with a language of finite words accepted by an OPA could be proved similarly as for classical automata if it were possible to recognize finite words by an OPA without emptying the stack and without even performing any flush move induced by symbol $\#$ immediately after reading the word; in this way the acceptance could be completed even when the words of the second language prevent emptying the stack.

To this aim, a possible solution is to introduce a variant of the semantics of the transition relation and of the acceptance condition for OPAs on finite-length words: a string is accepted if the automaton reaches a final state right at the end of the parsing of the whole word, and does not perform any flush move determined by the ending delimiter $\#$ to empty the stack; thus it stops just after having put the last symbol of $x$ on the stack.
Precisely, the semantics of the transition relation differs from the definition of
classical OPAs in that, once a configuration with the endmarker as lookahead is reached, the
computation cannot evolve in any subsequent configuration, i.e., a flush move 
$C \ \widetilde \vdash \ C_1$ with $C = \config { B_1 B_2 \dots B_n} {x\#}$ and $\symb {B_n} \gtrdot y\#$ is performed only if $y \neq \varepsilon$ (where symbol $\widetilde{\vdash}$ denotes a move according to this variant of the semantics of the transition relation). 
The language accepted by this variant of the automaton (denoted as $\widetilde{L}$) is the set of words:
\[
\widetilde L (\mathcal A) = \{ x \mid  \config {\stack 0\#{q_I}} {x\#}   \stackrel * {\widetilde{\vdash}}  
\config {\gamma \stack 0 a {q_F}} \#, q_I \in I, q_F \in F , \gamma \in \Gamma^*, a \in \Sigma \cup \{\#\} \}
\]

\noindent We emphasize that, unlike normal acceptance by final state of a pushdown automaton, which can perform a number of $\varepsilon$-moves after reaching the end of a string and accept if just one of the visited states is final, this type of automaton cannot perform any (flush) move after reaching the endmarker through the last look-ahead.

 Nevertheless, the variant and the classical definition of OPA are equivalent, as the following statements (Lemma~\ref{th:FAFW} and Statement 1) prove.

\begin{lemma}
\label{th:FAFW}
Let $\mathcal{A}_1$ be a nondeterministic OPA defined on an OP alphabet $(\Sigma, M)$ with $s$ states. Then there exists a nondeterministic OPA $\mathcal{A}_2$ with the same precedence matrix as $\mathcal{A}_1$ and $O(|\Sigma| s^2)$ states such that 
$L(\mathcal {A}_1) = \widetilde L(\mathcal {A}_2)$.

\end{lemma}

\noindent To build such a variant $\mathcal{A}_2$ we need some further notation.
Consider a word of finite length $w$ compatible with $M$: $\# w$ (without the closing \#).
Define a chain in a word $w$ as \textit{maximal} if it does not belong to a larger composed chain. In a word of finite length preceded and ended by $\#$ only the outmost chain $\chain{\#}{w}{\#}$ is maximal.

 An \textit{open chain} is a sequence of symbols $b_0 \lessdot a_1 \doteq a_2 \doteq \ldots \doteq a_n$, for $n\geq 1$.

\noindent The \textit{body} of a chain $\chain axb$, simple or composed, is the word $x$.
 A letter $a \in \Sigma$ in a word $\# w \#$ with $w \in \Sigma^*$ or $\# w$ with $w \in \Sigma^\omega$, where $w$ is compatible with $M$,  is \textit{pending} if it does not belong to the body of a chain, i.e., once pushed on the stack when it is read, it will never be flushed afterwards.

A word $w$ which is preceded but not ended by a delimiter $\#$ can be factored in a unique way as a sequence of bodies of maximal chains $w_i$ and pending letters $a_i$ as
$
\# \ w  = \# \ w_1 a_1 w_2 a_2 \ldots w_n a_n
$
where $\chain{a_{i-1}}{w_i}{a_{i}}$ are maximal chains and each $w_i$ can be possibly missing, with $a_0 = \#$ and $\forall i: 1\leq i \leq n-1$ $a_i\lessdot a_{i+1}$ or $a_i \doteq a_{i+1}$.

In general, during the parsing of word $w$, the symbols of the string are put on the stack and, whenever a chain is recognized, the letters of its body are flushed away. 

Hence, after the parsing of the whole word the stack contains only the symbols $\# \  a_1  \ a_2 \ldots \ a_n$ and is structured as a sequence of open chains.
Let $k$ be the number of open chains and denote by $a_1 = a_{i_1}, a_{i_2}, \dots a_{i_k}$ their starting symbols, then the stack contains:
\vspace{-0.1cm}
% $$
% \# \lessdot a_{1} \doteq a_{2} \doteq \ldots \doteq a_{k_1} \lessdot a_{k_1+1} \doteq \ldots \doteq a_{k_1+k_2} \lessdot \ldots \doteq a_{k_1+\ldots+k_h} \lessdot a_{k_1+\ldots+k_h+1} \doteq \ldots  a_{n=k_1+\ldots+k_{h+1} }
% $$
% where $k_j \geq 1, \forall j$.
$$
\# \lessdot a_{i_1} = a_1 \doteq a_{2} \doteq \ldots 
\lessdot a_{i_2} \doteq a_{i_2 +1} \ldots 
\lessdot a_{i_3} \doteq a_{i_3 +1} \ldots  
\lessdot a_{i_k} \doteq a_{i_k +1} \ldots  \doteq a_n
$$

%\noindent where $k$ is the number of open chains and we denoted by $a_1 = a_{i_1}, a_{i_2}, \dots a_{i_k}$ their starting symbols.
When a word $w$ is parsed by a classical OPA, the automaton performs a series of flush moves at the end of the string due to the presence of the final symbol $\#$. These moves progressively empty the stack, removing one by one the open chains and, for each such flush, they update the state of the automaton on the basis of the symbols which delimit the portion of the stack to be removed, which correspond to the state symbols at the end of the current open chain and at the end of the preceding open chain. The run is accepting if it leads to a final state after the flush moves.

 As an example, the transition sequence below shows the flush moves of a classical OPA when it reaches the position of $a_n$:\medskip\\
 $ \langle \stack{0}{\#}{q_1} \stack{1}{a_{i_1}}{q_2}	\stack{0}{a_2}{q_3} \ldots \stack{0}{a_{i_2 -1}}{q_{i_2}} \stack{1}{a_{i_2}}{q_{i_2+1}} \ldots \stack{0}{a_{i_3 -1}}{q_{i_3}} \ldots\stack{0}{a_{i_k -1}}{q_{i_k}}
 \stack{1}{a_{i_k}}{q_{i_k+1}}\ldots \stack{0}{a_n}{q_{n+1}}, \# \rangle\medskip\\$
 $\comp{} \langle \stack{0}{\#}{q_1} \stack{1}{a_{i_1}}{q_2}	\stack{0}{a_2}{q_3} \ldots \stack{0}{a_{i_2 -1}}{q_{i_2}} \stack{1}{a_{i_2}}{q_{i_2+1}} \ldots \stack{0}{a_{i_3 -1}}{q_{i_3}} \ldots
 \stack{0}{a_{i_k -1}}{\hat q_{i_k} = \delta_{\text{flush}}(q_{n +1}, q_{i_k})}, \# \rangle\medskip\\$
 $\comp{*} \langle \stack{0}{\#}{q_1} \stack{1}{a_{i_1}}{q_2}	\stack{0}{a_2}{q_3} \ldots \stack{0}{a_{i_2 -1}}{q_{i_2}} \stack{1}{a_{i_2}}{q_{i_2+1}} \ldots 
 \stack{0}{a_{i_3 -1}}{\hat q_{i_3} = \delta_{\text{flush}}(\hat q_{i_4}, q_{i_3})} , \# \rangle\medskip\\$
 $\comp{} \langle \stack{0}{\#}{q_1} \stack{1}{a_{i_1}}{q_2}	\stack{0}{a_2}{q_3} \ldots \stack{0}{a_{i_2 -1}}{\hat q_{i_2}= \delta_{\text{flush}}(\hat q_{i_3}, q_{i_2})} , \# \rangle\medskip\\$
 $\comp{} \langle \stack{0}{\#}{\hat q_1= \delta_{\text{flush}}(\hat q_{2}, q_{1})} , \# \rangle\\$

A nondeterministic automaton that, unlike classical OPAs, does not resort to the delimiter $\#$ for the recognition of a string may guess nondeterministically the ending point of each open chain on the stack and may guess how, in an accepting run, the states in these points of the stack would be updated if the final flush moves were progressively performed. The automaton must behave as if, at the same time, it simulates two snapshots of the accepting run of a classical OPA: a move during the parsing of the string and a step during the final flush transitions which will later on empty the stack, leading to a final state.
To this aim, the states of a classical OPA are augmented with an additional component to store the necessary information.

In the initial configuration, the symbol at the bottom of the stack comprises, along with an initial state $q$ of the original OPA $\mathcal{A}_1$, an additional state, say $q_F$, which represents a final state of $\mathcal{A}_1$. The additional component is propagated until the automaton nondeterministically identifies the first pending letter, which represents the beginning of the first open chain; at this time the component is updated with a new state chosen so that there exists a move from it in $\mathcal{A}_1$ that can flush and replace the state at the bottom of the stack with the final one $q_F$ (notice that if the beginning letter of the word is not a pending letter -- i.e., the prefix of the word is a maximal chain -- after completing the parsing of the chain, the initial state $q$ will be flushed and replaced on the bottom of the stack by a new state, say $r$, like in a classical OPA; in this case the last component added after reading the pending letter is chosen so that there exists a move in the graph of $\mathcal{A}_1$ that can flush and replace the state $r$ with $q_F$). 
Then, similarly, the additional component is propagated until the ending point of each open chain, until the conclusion of the parsing; while reading the pending letter that represents the beginning of the successive open chain
the automaton augments the new state on the stack with a placeholder chosen so that there is a flush move in $\mathcal{A}_1$ from it that can replace the state at the end of the previous open chain with the additional component previously stacked, thus allowing a backward path of flush moves from each ending point of an open chain to the previous one, up to the final state initially stacked.
If the forward path consisting of moves during the parsing of the string and this backward path of flush moves can consistently meet and be rejoined when the parsing of the input string stops, then they constitute an entire accepting run of the classical OPA.

A variant OPA $\mathcal{A}_2$ equivalent to a given OPA $\mathcal{A}_1$ thus may be defined so that, after reading each prefix of a word, it reaches a final state whenever, if the word were completed in that point with $\#$, $\mathcal{A}_1$ could reach an accepting state with a sequence of flush moves. In this way, $\mathcal{A}_2$ can guess in advance which words may eventually lead to an accepting state of $\mathcal{A}_1$, without having to wait until reading the delimiter $\#$ and to perform final flush~moves.

\begin{example}
Consider the computation of the OPA in Example~\ref{ex:db}. If we consider the input word of this computation without the ending marker $\#$, then the sequence of pending letters on the stack, after the automaton puts on the stack the last symbol $D$, is
$\# \lessdot \cup \ \lessdot \Join \ \lessdot \ \Join \lessdot \ \pi_{expr} \lessdot D$. 
There are five open chains with starting symbols $ \cup, \ \Join, \ \Join, \ \pi_{expr}, D  $,
hence the computation ends with five consecutive flush moves determined by the delimiter $\#$. 
The following figure shows the configuration  just before looking ahead at the symbol \#. The states (depicted within a box) at the end of the open chains are those placeholders that an equivalent variant OPA should guess in order to find in advance the last flush moves $q_1 = \fcolorbox{black}{white}{$q_1$}\cflush{q_0}{gray}
\fcolorbox{black}{white}{$q_1$} \cflush{q_0}{gray}
\fcolorbox{black}{white}{$q_1$} \cflush{q_1}{gray}
\fcolorbox{black}{white}{$q_1$} \cflush{q_1}{gray}
\fcolorbox{black}{white}{$q_1$} \cflush{q_1}{gray}
\fcolorbox{black}{white}{$q_1 \in F_1$}$ of the accepting run.
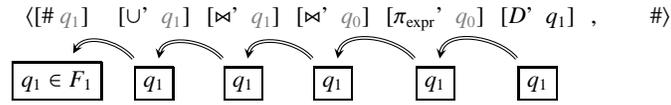
\vspace{-0.5cm}
\begin{figure}[h!]
\centering
\begin{tikzpicture}[flush/.style={double, >=stealth, thin, rounded corners}]
\matrix (m) [matrix of nodes]
 {
    $\langle$\stack{0}{\#}{$\color{gray} q_1$} & \stack{1}{$\cup$}{$\color{gray} q_1$}	& \stack{1}{$\Join$}{$\color{gray} q_1$}	& {\stack{1}{$\Join$}{$\color{gray} q_0$}} & {\stack{1}{$\pi_{\text{expr}}$}{$\color{gray}q_0$}} & {\stack{1}{$D$}{$q_1$}} &, &[0.5cm] $\# \rangle$ \\ 
    & \ & \ & \ & \ & \ & \ \\
    || $\fcolorbox{black}{white}{$q_1 \in F_1$}$ & |[]| $\fcolorbox{black}{white}{$q_1$}$	& |[]| $\fcolorbox{black}{white}{$q_1$}$ &|[]| $\fcolorbox{black}{white}{$q_1$}$ &|[]| $\fcolorbox{black}{white}{$q_1$}$ & |[]|$\fcolorbox{black}{white}{$q_1$}$ \\
 };

\draw[flush, ->] (m-3-6)  to [out=130, in=80] (m-3-5);
\draw[flush, ->] (m-3-5)  to [out=130, in=80] (m-3-4);
\draw[flush, ->] (m-3-4)  to [out=130, in=80] (m-3-3);
\draw[flush, ->] (m-3-3)  to [out=130, in=80] (m-3-2);
%\draw[flush, ->] (m-3-4)  to [out=170, in=70] (m-3-2);
\draw[flush, ->] (m-3-2)  to [out=130, in=60] (m-3-1);
\end{tikzpicture}\label{fig:newForm}
\end{figure}

%\newpage
\vspace{-0.5cm}
\noindent The corresponding configuration of the variant OPA, with the augmented states, would~be:
\begin{figure}[h!] %\vspace{-0.8cm} %\centering
\begin{tikzpicture}[flush/.style={double, >=stealth, thin, rounded corners}]

\matrix (m) [matrix of nodes]
 {
    $\langle$\stack{0}{\#}{$\textcolor{gray} {q_1}, \fcolorbox{black}{white}{$q_1$}$} & \stack{1}{$\cup$}{$\textcolor{gray} {q_1}, \fcolorbox{black}{white}{$q_1$}$}	& \stack{1}{$\Join$}{$\textcolor{gray} {q_1}, \fcolorbox{black}{white}{$q_1$}$}	& {\stack{1}{$\Join$}{$\textcolor{gray} {q_0}, \fcolorbox{black}{white}{$q_1$}$}} & {\stack{1}{$\pi_{\text{expr}}$}{$\textcolor{gray}{q_0}, \fcolorbox{black}{white}{$q_1$}$}} & {\stack{1}{$D$}{$q_1, \fcolorbox{black}{white}{$q_1$}$}} &, & $\# \rangle$ \\ 
 };

\end{tikzpicture}
\end{figure}
\vspace{-0.5cm}

\end{example}

\noindent We are now ready to formally prove Lemma~\ref{th:FAFW}.
\begin{proof}
\noindent 
%We now give a formal definition of the variant automaton $\mathcal{A}_2$ and we prove its equivalence with a classical OPA $\mathcal{A}_1$. 
 Let $\mathcal{A}_1 = \langle \Sigma, M, Q_1, I_1, F_1, \delta_1  \rangle$ and define $\mathcal{A}_2  = \langle \Sigma, M, Q_2, I_2, F_2, \delta_2  \rangle$ as follows.
\begin{itemize}
\item $Q_2 = \{B, Z, U \} \times \hat{\Sigma} \times Q_1 \times Q_1 $, where $\hat{\Sigma} = \Sigma \cup \{\#\}$.\\
Hence, a state $\langle x, a, q, p \rangle $ of $\mathcal{A}_2$ is a tuple whose first component 
denotes a nondeterministic guess for the symbol following the one currently read, i.e., whether it is
a pending letter which is the initial symbol of an open chain ($Z$), or a pending letter within an open chain ($U$),
or a symbol within a maximal chain ($B$). 
%or $\varepsilon$ when the input is over ($F$). 
%Actually, if the first component is $Z$, we can consider only tuple $\langle Z, a, q, p \rangle $ such that there exists $r \in \mathcal A_1 with $r \flush q p$.
The second and third components of a state represent, respectively, the lookback letter $a$ read to reach the state, 
and the current state $q$ in $\mathcal{A}_1$.
To illustrate the meaning of the last component, 
consider an accepting run of $\mathcal A_1$ 
and let $q$ be the current state just before a mark move is going to be performed 
at the beginning of an open chain; 
also let $r$ be the state reached by the mark move and $s$ be the state on top of the stack 
when this open chain is to be flushed replacing $q$ with a new state $p$. 
%$q$ will be updated as a new state $p$ by a flush move of $\mathcal A_1$ determined by the ending marker $\#$. 
Then, in the same position of the corresponding run of $\mathcal A_2$, 
the current state would be $\langle Z, a, q, p \rangle \in Q_2$ and state 
$\langle x, a, r, s \rangle \in Q_2$ will be reached 
by $\mathcal A_2$ ($x$ being nondeterministically anyone of $B$, $Z$, $U$), i.e., the last component $p$ represents a guess about the state that will replace $q$ 
in $\mathcal A_1$ when the starting open chain will be flushed.
 Hence we can consider only states  $\langle Z, a, q, p \rangle \in Q_2$ such that 
$s \flush q p$ in $\mathcal A_1$ for some $s \in Q_1$.
In all other positions the last component of the states in $Q_2$ is simply propagated.

\item $I_2 = \{\langle  x, \#, q, q_F \rangle \mid x \in \{Z, B\}, q \in I_1, q_F \in F_1\}$
% to accept the empty string\\$ \cup \  \
% \langle  Z, \#, q, b, q_f \rangle \mid  q \in I_1 \wedge \exists \text{ an edge in $\mathcal{A}_1$ } p \flush{q} q_f,  \ q_f \in F_1, p \in Q_1, b \in \Sigma \}\\ \cup \{\langle  \bot, \#, q, b, q_f \rangle \mid  q \in I_1, q_f \in F_1, b \in \Sigma\}$\smallskip
\item $F_2 = \{\langle  Z, a , q, q \rangle \mid  q \in Q_1, a \in \hat \Sigma\}$\smallskip
\item The transition function is defined as the union of two disjoint functions.
% $\delta_2: Q_2 \times (\Sigma \cup Q_2) \rightarrow 2^{Q_2}$ %(the presence of the various cases exhibited next will be justified in the following proof of equivalence of the languages recognized by the two automata).

The push transition function $\delta_{\text{2push}} : Q_2 \times \Sigma \rightarrow 2^{Q_2} $ is defined as follows,
where $p,q,r,s \in Q_1$, $a \in \hat\Sigma$, and $b,c \in \Sigma$.

\begin{itemize}

\item \emph{Mark of a pending letter at the beginning of an open chain.} 
If $a \lessdot b$ then:
\[
\delta_{\text{2push}} \left( \langle Z, a, q, p \rangle, b \right) =
\left\{ 
\langle x, b, r, s \rangle  \mid 
x \in \{ B, Z, U \},
q \va b r, s \flush q p \text{ in } \mathcal A_1
\right\}
\]

\item \emph{Push of a pending letter within an open chain.} 
If $a \doteq b$ then:
\[
\delta_{\text{2push}} \left( \langle U, a, q, p \rangle, b \right) =
\left\{ 
\langle x, b, r, p \rangle  \mid 
x \in \{ B, Z, U \},
q \va b r  \text{ in } \mathcal A_1
\right\}
\]

\item \emph{Push/mark of a symbol of a maximal chain.} 
\[
\delta_{\text{2push}} \left( \langle B, a, q, p \rangle, b \right) =
\left\{ 
\langle B, b, r, p \rangle  \mid q \va b r \text{ in } \mathcal A_1 
\right\}
\]

\end{itemize}
Notice that the second and third components of the states computed by $\delta_{\text{2push}}$ 
are independent of the first component of the starting state.

The flush transition function $\delta_{\text{2flush}} : Q_2 \times Q_2 \rightarrow 2^{Q_2} $ 
can be executed only within a maximal chain since there are no flush determined by the ending delimiter:

%\item \emph{Flush within a maximal chain.} 
\[
\delta_{\text{2flush}} \left( \langle B, b, q, s \rangle, \langle B, c, p, s \rangle \right) =
\left\{ 
\langle x, c, r, s \rangle  \mid 
x \in \{ B, Z, U \},
q \flush p r \text{ in } \mathcal A_1
\right\}
\]

All other moves lead to an error state.
%Notice that the moves $\delta_{\text{2flush}}(\langle B, \ldots\rangle, \langle Z, \ldots\rangle) $, $\delta_{\text{2flush}}(\langle\bot, \ldots\rangle, \langle\ldots\rangle) $,\\ $\delta_{\text{2flush}}(\langle Z, \ldots\rangle, \langle  \ldots\rangle) $ lead to an error state.
\end{itemize}

\noindent The automata $\mathcal{A}_1$ and $\mathcal{A}_2$ recognize the same language, $L(\mathcal{A}_1) = \widetilde L(\mathcal{A}_2)$.

Let us prove first $L(\mathcal{A}_1) \subseteq \widetilde L(\mathcal{A}_2)$.
Let $w \in L(\mathcal{A}_1) $ be a finite-length word. 
Then there exist a support $q \ourpath{w}{q'}$ in $\mathcal {A}_1$ with $q \in I_1$ and $q' \in F_1$.
If $ w = w_1 a_1 w_2 a_2 \dots w_n a_n \in L(\mathcal{A}_1) $ where $a_i$ are pending letters and $w_i$
are maximal chains, let $k$ be the number of open chains that remain on the stack after the parsing of the last symbol in $\Sigma$ of $w$, and let $a_{i_1}=a_1, a_{i_2}, \dots, a_{i_k}$ be their starting symbols. 
%Notice that, except for $t=1$, in general $a_{i_t}$ does not coincide with $a_t$.
Also, for every $i=2, \dots, n$, let $t(i)$ be the greatest index $t$ such that $i_t < i$, i.e., 
$a_i$ is within the $t(i)$-th open chain starting with $a_{i_{t(i)}}$.
In particular, for $i = n$, if $a_{n-1} \lessdot a_n$ then $i_k = n$, otherwise $t(n) = k$.

Then the above support for $w$ can be decomposed as
\begin{equation}
\label{eq:path-concat}
q = \widetilde q_0 
\ourpath{w_1}{q_1} \va{a_1}{\widetilde q_1}
\ourpath{w_2}{q_2} \va{a_2}{\dots} \ourpath{w_n}{q_n}  \va{a_n}{\widetilde{q_n}} = p_k
\end{equation}
\[ 
\widetilde q_n = p_k \flush{q_{i_k}}{p_{k-1}} 
\flush{q_{i_{k-1}}}{p_{k-2}} 
\flush{}{\dots} \flush {}{p_2} 
\flush{q_{i_2}}{p_1} 
\flush{q_{i_1}=q_1}{p_0} = q' 
\]
where $q_i = \widetilde q_{i-1}$ if $w_i = \varepsilon$ for $i=1,2,\dots, n$.
Notice that, for every $t$, $q_{i_t}$ is the state reached in this path before the mark move that pushes symbol $a_{i_t}$ on the stack;
%while the flush move labeled by $q_{i_t}$ remove from the stack the open chain starting with symbol $a_{i_t}$ replacing state 
%$q_{i_t}$ with $p_{t-1}$ on top on the stack.
moreover, when the open chain starting with $a_{i_t}$ is to be flushed, the current state is $p_t$ 
and then state $q_{i_t}$ is replaced with $p_{t-1}$ on top of the stack.

%Thus, by composing in the right order the previous computations and 
Starting with state $\langle Z, \#, q_1, p_0 \rangle $ if $w_1 = \varepsilon$ or with
$\langle B, \#, \widetilde q_0, p_0 \rangle \ourpath {w_1}{\langle Z, \#, q_1, p_0 \rangle} $
if $w_1 \neq \varepsilon$, an accepting computation of $\mathcal A_2$ can be built on the basis of the following facts:
\begin{itemize}
\item Since 
$q_1 \va{a_1}{ \widetilde q_1}$ and $p_1 \flush {q_1}{p_0}$
in $\mathcal A_1$, then
$
\delta_{2\text{push}}(\langle Z, \#, q_1, p_0 \rangle, a_1) \ni \langle x, a_1, \widetilde q_1, p_1 \rangle 
$
in  $\mathcal A_2$ for $x \in \{U,Z\}$.
This is a mark move that can be applied at the beginning of the first open chain starting with $a_1$, where
 $p_1$ is the guess about the state that will be reached before such open chain will be flushed.

\item In general, for every $t$, since 
$q_{i_t} \va{a_{i_t}}{ \widetilde q_{i_t}}$ and $p_t \flush {q_{i_t}}{p_{t-1}}$
in $\mathcal A_1$, then\\
$
\delta_2(\langle Z, a_{i_{t}-1}, q_{i_t}, p_{t-1} \rangle, a_{i_t}) \ni \langle x, a_{i_t}, \widetilde q_{i_t}, p_t \rangle 
$
for $x \in \{U,Z\}$.
This is a mark move that can be applied at the beginning of the $t$-th open chain starting with $a_{i_t}$, where
$p_t$ is the guess about the state that will be reached before such open chain will be flushed.
In particular, if $i_k = n$, we can reach state
$\langle Z, a_n, \widetilde q_n, p_k \rangle$ which is final in $\mathcal A_2$ since $q_n = p_k$.

\item For every maximal chain $w_i$ of $w$ (with $i \geq 2$) consider its support $\va{a_{i-1}}{\widetilde q_{i-1}} \ourpath{w_i}{q_i}$ in~\eqref{eq:path-concat}.
Then in $\mathcal A_2$ we have the sequence of moves ``summarized'' (with a natural overloading of the notation) by
$
\delta_2 \left( \langle B, a_{i-1}, \widetilde q_{i-1}, p_{t(i)} \rangle, w_i \right)
\ni \langle x, a_{i-1}, q_i, p_{t(i)} \rangle
$,
where $x \in \{U,Z\}$.
%where $\ell_i$ is the greater index $i_t$ such that $a_{i_t} < i$.
Notice that the last component of the states does not change because we are within a maximal chain. In particular, 
during the parsing of $w_i$
the last component is equal to $p_{t(i)}$, as guessed by the mark move %of symbol $a_{i_{t(i)}}$ 
at the beginning of the current open chain.

\item 
For every $i\not\in\{i_1, i_2, \dots, i_k\}$, since $\delta_{1\text{push}} (q_i, a_i) \ni \widetilde q_i$, then
$\delta_{2\text{push}} ( \langle U, a_{i-1}, q_i, p_{t(i)} \rangle, a_i ) $ 
contains $\langle x, a_i, \widetilde q_i, p_{t(i)} \rangle$,
for $x \in \{B,Z,U\}$.
In particular, if $n \neq i_k$, then $t(n) = k$ and for $i = n$ we can reach state
$\langle Z, a_n, \widetilde q_n, p_k \rangle$ which is final in $\mathcal A_2$ since $q_n = p_k$.

\end{itemize}
Thus, 
by composing in the right order the previous moves,
% and starting with 
%$\langle Z, \#, q_1, p_0 \rangle $ if $w_1 = \varepsilon$ or with
%$\langle B, \#, \widetilde q_0, p_0 \rangle \ourpath {w_1}{\langle Z, \#, q_1, p_0 \rangle} $
%if $w_1 \neq \varepsilon$, 
one can obtain an accepting computation for $w$ in $\mathcal A_2$.

Conversely, to prove that $\widetilde L(\mathcal{A}_2) \subseteq L(\mathcal{A}_1)$, consider a finite word $w \in \widetilde L(\mathcal{A}_2)$. Then there exists a successful run of $\mathcal{A}_2$ on $w$.
Let $w$ be factorized as above;
then the accepting run for $w$ can be decomposed as
\[
\pi_0 \ourpath{w_1}{\rho_1} \va{a_1}{\pi_1} \ourpath{w_2}{\rho_2} 
\dots
\rho_i \va{a_i}{\pi_i} \ourpath{w_{i+1}}{}
\dots 
\ourpath{w_{n}}{\rho_n} \va{a_n}{\pi_n} 
\]
where $\pi_i, \rho_i \in Q_2$, 
$\rho_i = \pi_{i-1}$ if $w_i = \varepsilon$, 
$\pi_0 \in I_2$ and $\pi_n \in F_2$.
% 
% Let $k$ be the number of open chains that remains on the stack after the parsing of the last symbol in $\Sigma$ of $w$, and $a_1 = a_{i_1}, a_{i_2}, \dots a_{i_k}$ be their starting symbols.
By projecting this path on the third component of states $\pi_i$ and $\rho_i$ (given by, say, $p_i$ and $r_i \in Q_1$),
 we obtain a path in $\mathcal A_1$
labelled by $w$. This path is not accepting because there are open chains left on the stack that need flushing,
but we can complete this path arguing by induction on the structure of maximal chains 
according to the definition of $\delta_2$.
% such that the states in $Q_1$ involved in the run, 
%before the final flush moves determined by the delimiter $\#$,
%are exactly the projections, say $p_i$ and $r_i \in Q_1$, on the third component of states $\pi_i$ and $\rho_i$.
More formally,  
one can verify that $Q_1$ contains suitable states
$p_i$ (for $0 \leq i \leq n)$, 
$r_i$ (for $1 \leq i \leq n)$, 
$s_t$ (for $1 \leq t \leq k)$, 
with $r_i = p_{i-1}$ whenever $w_i = \varepsilon$,
such that the following facts hold.
\begin{itemize}

\item 
$\pi_0 \in I_2$, hence $\pi_0 = \langle x_0, \#, p_0, s_0 \rangle$,
with $p_0 \in I_1$ and $s_0 \in F_1$;  
$x_0$ is $B$ if $w_1 \neq \varepsilon$, otherwise $x_0 = Z$. 

\item 
$\pi_0 \ourpath{w_1}{\rho_1}$ in $\mathcal A_2$ implies that the last component of state $\pi_0$
is propagated through chain $w_1$ without change; hence
$\rho_1 = \langle Z, \#, r_1, s_0 \rangle$ 
with $p_0 \ourpath{w_1}{r_1}$ in $\mathcal A_1$. 

\item 
$\rho_1 \va{a_1}{\pi_1}$ is a mark move of $\mathcal A_2$ at the beginning of an open chain,
and this implies that the last component of $\pi_1$ is new; hence we have
$\pi_1 = \langle x_1, a_1, p_1, s_1 \rangle$ 
with $r_1 \va{a_1}{p_1}$ and $s_1 \flush{r_1}{s_0}$ in $\mathcal A_1$;
the first component is $x_1 = B$ if $w_2 \neq \varepsilon$ otherwise $x_1$ equals $Z$ or $U$
according to whether $a_2$ starts an open chains or not, respectively,

\item 
The flush moves within $\pi_i \ourpath{w_{i+1}}{\rho_{i+1}}$ for $1 \leq i < i_2$,
and the push moves within an open chain $\rho_i \va{a_i}{\pi_i}$ for $1 < i < i_2$
propagate with no change the last component of states. Hence  
$\rho_i =  \langle U, a_{i-1}, r_i, s_1 \rangle$ 
and $\pi_i = \langle x_i, a_i, p_i, s_1 \rangle$ 
with $p_{i-1} \ourpath{w_i}{r_i}\va{a_i}{p_i}$ in $\mathcal A_1$.
The first component is 
$x_i = B$ if $w_i \neq \varepsilon$ otherwise $x_i = Z$ for $i = i_2-1$ and $x_i = U$ in the other cases.

\item 
$\rho_{i_2} \va{a_{i_2}}{\pi_{i_2}}$ is a mark move of $\mathcal A_2$ at the beginning of an open chain,
and this implies that the last component of $\pi_1$ is new; hence we have
$\pi_{i_2}  = \langle x_{i_2} a_{i_2}, p_{i_2}, s_2 \rangle$ 
with $r_{i_2} \va{a_{i_2}}{p_{i_2}}$ and $s_2 \flush{r_{i_2}}{s_1}$ in $\mathcal A_1$. 
The first component is $x_{i_2} = B$ if $w_{i_2} \neq \varepsilon$ otherwise $x_1$ equals $Z$ or $U$
according to whether $a_{i_2}+1$ starts an open chains or not, respectively.

\item Similarly for the following moves in the run.
\end{itemize}
In general, we get %(for $y_i \in \{Z,U\}, x_i \in \{B, Z, U\}$):
\begin{eqnarray*}
\rho_i = \langle y_i, a_{i-1}, r_i, s_{t(i)} \rangle &\qquad& \text{for every } i = 1, 2, \dots, n, \\
%where $\ell_i$ is the greater index $i_t < i$,
%
\pi_i  =\langle x_i, a_i, p_i, s_{t(i)} \rangle &\qquad& \text{for every } i \not \in \{i_1, i_2, \dots, i_k\}, \\
\pi_{i_t} =  \langle x_{i_t}, a_{i_t}, p_{i_t}, s_t \rangle &\qquad& \text{for every } t =1, 2, \dots, k,\\
\text{with } r_i \va{a_i}{p_i}, \  s_t \flush{r_{i_t}}{s_{t-1}}, \ p_{i-1} \ourpath{w_i}{r_i}
\text{ in } \mathcal A_1\\
\text{and }y_i \in \{Z,U\}, x_i \in \{B, Z, U\}
&\qquad& \text{for every } i \text{ and } t.
\end{eqnarray*}
%
% with 
% \[
% r_i \va{a_i}{p_i}, \qquad    s_t \flush{r_{i_t}}{s_{t-1}}, \qquad   p_{i-1} \ourpath{w_i}{r_i}
% \]
By convention, $a_0 = \#$.
For $i = n$ we have $n = i_k$ or $ t(n) = k$, hence 
$\pi_n =  \langle x_n, a_n, p_n, s_k \rangle$, and  $p_n = s_k$ and $x_n = Z$ since $\pi_n \in F_2$.
%recall that $i_k = n$ or $ t(n) = k$, hence 
% hence $p_n = s_k$:
%if $i_k = n$, i.e., the last pending letter starts an open chain, then   
\noindent Thus, in $\mathcal A_1$ there is an accepting run
\[
I_1 \ni p_0 \ourpath{w_1}{r_1} \va{a_1}{p_1} \ourpath{w_2}{r_2} 
\dots
r_i \va{a_i}{p_i} \ourpath{w_{i+1}}{}
\dots 
\ourpath{w_{n}}{r_n} \va{a_n}{p_n} = s_k
\]
\[ 
p_n = s_k \flush{r_{i_k}}{s_{k-1}} 
\flush{r_{i_{k-1}}}{s_{k-2}} 
\flush{}{\dots} \flush {}{s_2} 
\flush{r_{i_2}}{s_1} 
\flush{r_{i_1}=r_1}{s_0} \in F_1
\] 
and this concludes the proof of the lemma.
\hfill\qed

\end{proof}

The next Statement, although not necessary to prove closure under concatenation of \lof{\bfa}, completes the proof of equivalence between traditional and variant OPAs, showing how to define, for any variant OPA, a classical OPA which recognizes the same language.
\begin{paragraph}{\textnormal{\textbf{Statement 1}}}
\label{th:FWFA}
Let $\mathcal{A}_2$ be a nondeterministic OPA defined on an OP alphabet $(\Sigma, M)$ with $s$ states.
% and $L = \widetilde  L(\mathcal {A}_2) \subseteq \Sigma^*$.
Then there exists a nondeterministic OPA $\mathcal{A}_1$ with the same precedence matrix as $\mathcal{A}_2$ and $O(|\Sigma|^2 s)$ states such that $L(\mathcal {A}_1) = \widetilde L (\mathcal A_2)$.
\end{paragraph}
\begin{proof}
Let $\mathcal{A}_2 = \langle \Sigma, M, Q, I, F, \delta  \rangle$ and consider, first, an equivalent form for the automaton $\mathcal{A}_2$, where all the states are simply enriched with a lookahead and lookback symbol:
$\mathcal{\tilde{A}}_2 = \langle \Sigma, M, Q_2, I_2, F_2, \delta_2  \rangle$ where

\begin{itemize}
\item $Q_2 = \hat \Sigma \times Q \times \hat \Sigma$, where $\hat{\Sigma} = (\Sigma \cup \left\{ \# \right\} )$, i.e. the first component of a state is the lookback symbol, the second component of the triple is a state of $\mathcal{A}_2$ and the third component of the state is the lookahead symbol,
\item $I_2 = \{\#\} \times I \times \{a \in \hat \Sigma \mid M_{\# a} \neq \emptyset\}$ is the set of initial states of $\mathcal {\tilde{A}}_2$,
\item $F_2 = (\{\#\} \cup \{b \in \Sigma: b \gtrdot \#\} ) \times F \times \{\#\}$
\item and the transition function $\delta_2 : Q_2 \times ( \Sigma \cup Q_2) \rightarrow 2^{Q_2}$ is defined in the following natural way
\begin{itemize}
\item $\delta_{\text{2push}}(\langle  a, q, b \rangle, b) =\{ \langle b, p, c \rangle \mid p \in \delta_{\text{push}}(q,b) \wedge M_{a b} \in \{ \lessdot, \doteq \}  \wedge M_{b c} \neq \emptyset \},$\smallskip\\
$\forall a \in \hat \Sigma, b \in \Sigma, q \in Q$
\item $\delta_{\text{2flush}}( \langle a_1, q_1, a_2 \rangle, \langle b_1, q_2, b_2 \rangle) =\{\langle b_1, q_3, a_2 \rangle \mid q_3 \in \delta_{\text{flush}}(q_1,q_2) \text{ } \wedge \ M_{a_1 a_2} = \gtrdot$\smallskip\\ $ \wedge \ M_{b_1 a_2} \neq \emptyset\},$\smallskip\\
$\forall a_1, a_2, b_2 \in \Sigma, \forall b_1 \in \hat \Sigma, \forall q_1, q_2 \in Q$.
\end{itemize}
\end{itemize}

It is clear that $\widetilde L(\mathcal{A}_2) = \widetilde L(\mathcal{\tilde{A}}_2)$. Furthermore, the final states of $\mathcal{\tilde{A}}_2$ cannot be reached by flush edges: in fact, if there exists a transition $\triple{a_1}{q_1}{a_2} \lflush{\triple{b_1}{q_2}{b_2}} \triple{a_1}{q_3}{\#}$ towards a final state $\triple{a_1}{q_3}{\#}$, then the third component of the flushed and of the reached final state must be equal by definition of the transition function, i.e $ \triple{a_1}{q_1}{a_2} = \triple{a_1}{q_1}{\#}$. But this flush transition cannot be performed by a variant OPA, which stops a computation right before reading the delimiter $\#$, when the parsing of the word ends.

Hence, one may always refer to a variant OPA assuming that in its graph there are no flush moves towards final states.

It is then possible to describe an automaton OPA $\mathcal{A}_1$ equivalent to the variant OPA $\mathcal{A}_2$ (or $\mathcal{\tilde{A}}_2$). 

$\mathcal{A}_1 = \langle \Sigma, M, Q_1, I_1, F_1, \delta_1  \rangle$ is defined as $\mathcal{\tilde{A}}_2$ but it is enriched with an additional state, which is the only final state of $\mathcal{A}_1$ and which is reachable through a flush edge by all final states of $\mathcal{\tilde{A}}_2$. Basically, its role is to let $\mathcal{A}_1$ empty the stack after parsing a word that is accepted by $\mathcal{\tilde{A}}_2$. 
\begin{itemize}
\item $Q_1 = Q_2 \cup \{q_\text{accept}\}$
\item $I_1 = I_2  \cup \{q_\text{accept}\} \text{ if } I_2 \cap F_2 \neq \emptyset$ or $I_1 = I_2$ otherwise
\item $F_1 = \{q_\text{accept}\}$
\item The transition function $\delta_1$ equals $\delta_2$ on all states in $Q_2$; in addition $\mathcal{A}_1$ has departing flush edges from the final states in $F_2$ to $q_\text{accept}$ and $q_\text{accept}$ has no outgoing push/mark edge but only self-loops flush edges.\\
The push transition function $\delta_{\text{1push}}: Q_1 \times \Sigma \rightarrow 2^{Q_1}$ is defined as $\delta_{\text{1push}}(q,c) = \delta_{\text{2push}}(q,c), \forall q \in Q_2, c \in \hat{\Sigma}$, whereas $\delta_{\text{1push}}(q_\text{accept},c)$ leads to an error state for any $c$.\\
The flush transition $\delta_{\text{1flush}}: Q_1 \times Q_1 \rightarrow 2^{Q_1}$ is defined by:\smallskip\\
$\delta_{\text{1flush}}(q,p) = \delta_{\text{2flush}}(q,p), \forall q, p \in Q_2$\smallskip\\
$\delta_{\text{1flush}}(q, p) = q_\text{accept}, \forall q \in (F_2 \cup \{q_\text{accept}\}) , p \in Q_2$
\end{itemize}
The two automata recognize the same language, $L(\mathcal{A}_1) =\widetilde L(\mathcal{\tilde{A}}_2)$.

First of all, $L(\mathcal{A}_1) \subseteq \widetilde L(\mathcal{\tilde{A}}_2)$: in fact, if the OPA $\mathcal{A}_1$ recognizes a word, then it is either the empty word and thus $q_\text{accept} \in I_1$ and also $\mathcal{\tilde{A}}_2$ has a successful run on it, or $\mathcal{A}_1$ recognizes a word $w \neq \varepsilon$ and there exists a run $S$ of $\mathcal{A}_1$ which ends in the final state $q_\text{accept}$, emptying the stack. Notice that $q_\text{accept}$ is reached by a flush move from a state in $F_2$, say $q_f \in F_2$:
$$S: q_0 \in I_2 \ourpath{w} q_f \flush{ } q_\text{accept} (\lflush{p \in Q_1} q_\text{accept})^*$$
and $q_f$ itself is reached exactly when the parsing of the word $w$ is finished, since, as said before, a state in $F_2$ cannot be reached by flush moves. This condition is necessary to avoid the presence of sequences of flush moves from non accepting states towards final states.
Then the path from $q_0$ to $q_f$, which follows the same state and edges as $S$, represents a run of $\mathcal{\tilde{A}}_2$ which ends in a final state $q_f$ right after the parsing of the whole word, thus  accepting $w$.
The direction from right to left $L(\mathcal{A}_1) \supseteq \widetilde L(\mathcal{\tilde{A}}_2)$ derives easily from the fact that, if $\mathcal{\tilde{A}}_2$ accepts a word along a successful run, then $\mathcal{A}_1$ recognizes the word along the same run, possibly emptying the stack in the final state $q_\text{accept}$.
$\hfill\qed$
\end{proof}

Given the variant for OPAs on finite words, it is possible to prove the closure under concatenation of the class of languages accepted by \bfa s with a language of finite words accepted by an OPA, as the following theorem (Theorem~\ref{th:clConcatBFAF}) states. Notice that its proof differs from the non-trivial proof of closure under concatenation of OPLs of finite-length words~\cite{Crespi-ReghizziM12}, which, instead, can be recognized deterministically.

\begin{theorem} %[Closure of $\mathcal{L}$(\bfa) under concatenation]
\label{th:clConcatBFAF}
Let $L_1 \subseteq \Sigma^*$ be a language of finite words recognized by an OPA with OPM $M_1$ and $s_1$ states.
Let $L_2 \subseteq \Sigma^\omega$ be an $\omega$-language recognized by a nondeterministic \bfa\ with OPM $M_2$ compatible with $M_1$ and $s_2$ states.\\
Then the concatenation $L_1 \cdot L_2$ is also recognized by an \bfa\ with OPM $M_3\supseteq M_1 \cup M_2$ and $O(|\Sigma|( s_{1}^2 +  s_{2}^2))$ states.
\end{theorem}

\begin{proof}
Let $\mathcal{A}_1 = \langle \Sigma, M_1, Q_1, I_1, F_1, \delta_1  \rangle$  be a nondeterministic OPA which recognizes language $L_1$ and let  $\mathcal{A}_2 = \langle \Sigma, M_2, Q_2, I_2, F_2, \delta_2  \rangle$ be a nondeterministic \bfa\ with OPM $M_2$ compatible with $M_1$ which accepts $L_2$. Suppose, without loss of generality, that $Q_1$ and $Q_2$ are disjoint.

To define an automaton \bfa\ $\mathcal{A}_3$ which accepts the language $L_1 \cdot L_2$, we first build an automaton OPA in the variant form  $\mathcal{A'}_1 = \langle \Sigma, M_1, Q'_1, I'_1, F'_1, \delta'_1  \rangle$ such that $\widetilde L(\mathcal{A'}_1) = L(\mathcal{A}_1)$.

The automaton $\mathcal{A}_3$ may recognize the first finite words in the concatenation $L_1 \cdot 
L_2$ simulating $\mathcal{A'}_1$: during the parsing of the input string, if $\mathcal{A'}_1$ reaches a final state at the end of a finite-length prefix, then it belongs to $L_1$ and  $\mathcal{A}_3$ may immediately start the recognition of the second infinite string without the need to perform any flush move to empty the stack.
From this point onwards, then, $\mathcal{A}_3$ may check that the remaining infinite portion of the input belongs to $L_2$, behaving as the \bfa\ $\mathcal{A}_2$. Notice, however, that as it happens for operator precedence languages of finite-length words~\cite{Crespi-ReghizziM12}, the strings of the concatenation of two OPLs may have syntax trees that significantly differ from the concatenation of the trees of the single words: the trees of the strings of the two languages may be merged, according to the precedence relations between the symbols of the words, in a completely new structure. From the point of view of the parsing of a string in $L_1 \cdot L_2$ by an automaton, the joining of the trees of two words in $L_1$ and $L_2$ may imply that the recognition and reduction by flush moves of a subtree with branches in a word in $L_1$ have to be postponed until the parsing of the other branches in the word in $L_2$ has been completed.
Therefore, $\mathcal{A}_3$ cannot merely read the second infinite word performing the same transitions as $\mathcal{A}_2$, but it is still possible to simulate this \bfa\ keeping in the states some summary information about its runs. In this way, while reading the second word in the concatenation, whenever $\mathcal{A}_3$ has to reduce a subtree which extends to the previous word in $L_1$ and thus it has to perform a flush move that involves the portion of the stack piled up during the parsing of the first word, it can still restore on the stack the state that $\mathcal{A}_2$ would instead have reached, resuming the parsing of the second word thereon as in a run of $\mathcal{A}_2$.

In particular, the automaton $\mathcal{A}_3$ is defined as follows. Let $\hat{\Sigma}$ be $ \Sigma \cup \{\#\}$ and $\mathcal{A}_3 = \langle \Sigma, M_3, Q_3, I_3, F_3, \delta_3  \rangle$ where:
\begin{itemize}
\item $M_3 \supseteq M_1 \cup M_2$ and may be supposed to be a total matrix, for instance assigning arbitrary precedence relations to the empty entries, so that the strings in the concatenation of languages $L_1$ and $L_2$ are compatible with $M_3$.
\item $Q_3 = Q'_1 \ \cup \ \hat{\Sigma}\times Q_2 \times (Q_2 \cup \{-\})$, i.e. the set of states of $\mathcal{A}_3$ includes the states of $\mathcal{A'}_1$, while the states of $\mathcal{A}_2$ are extended with two components. The first component is a lookback symbol, the second component is the state of $Q_2$ that would be reached by $\mathcal{A}_2$ during its corresponding computation, and the third represents, as in the construction for deterministic OPAs~\cite{LonatiMandrioliPradella2011b}), the state with the marked symbol that, when the current input letter is read in a run performed by $\mathcal{A}_2$ on the infinite substring, is the last marked symbol on the stack. Storing this component is necessary to guarantee that, whenever the automaton $\mathcal{A}_3$ has to perform a flush move towards states piled in the stack during the recognition of the first word in the concatenation, it is still possible to compute the state that $\mathcal{A}_2$ would have reached instead.\\ This third component is denoted $'-'$ if all the preceding symbols in the stack have been piled during the parsing of the first word of the concatenation (thus the stack of $\mathcal{A}_2$ is empty).
\item $I_3 = I'_1 \cup \{\langle  \#, p_0, - \rangle \mid p_0 \in I_2\}$ if $\varepsilon \in L_1$ or $I_3 = I'_1$ otherwise
\item $F_3 = \hat{\Sigma} \times F_2 \times Q_2$
\item The transition function $\delta_3: Q_3 \times ( \Sigma \cup Q_3) \rightarrow 2^{Q_3}$ is defined as follows.
The push transition $\delta_{\text{3push}}: Q_3 \times \Sigma \rightarrow 2^{Q_3}$ is defined by:
\begin{itemize}
%\item $\delta_{\text{3push}}(q_1, c) = \delta'_{\text{1push}}(q_1, c), \forall q_1 \in Q'_1, c \in \Sigma$, i.e. it simulates $\mathcal{A'}_1$ on $Q'_1$
%\item $\delta_{\text{3push}}(q_1, c) = \{\langle  \#, p_0, -\rangle \mid p_0 \in I_2\}, \forall q_1 \in Q'_1, c \in \Sigma: \exists q_f \in F'_1 \text{ s.t. }\smallskip\\ \delta'_{\text{1push}}(q_1,c) \ni q_f$,\\ i.e. it reaches the initial states of $\mathcal{A}_2$ after the recognition of a word in $L_1$
\item $\delta_{\text{3push}}(q_1, c) = \delta'_{\text{1push}}(q_1, c) \cup \{\langle  \#, p_0, -\rangle \mid p_0 \in I_2, \text{ if }\exists q_f \in F'_1 \text{ s.t. } \delta'_{\text{1push}}(q_1,c) \ni q_f\}$, $ \forall q_1 \in Q'_1, c \in \Sigma$, 

i.e., it simulates $\mathcal{A'}_1$ on $Q'_1$ or nondeterministically enters the initial states of $\mathcal{A}_2$ after the recognition of a word in $L_1$
\item $\delta_{\text{3push}}(\langle a, p, r \rangle, c) = \left\{  
\begin{array}{ll}
\{\langle c, q, p \rangle \mid q \in \delta_{\text{2push}}(p, c)\} & \text{if } a \lessdot c \\
\{\langle c, q, r \rangle \mid q \in \delta_{\text{2push}}(p, c)\} & \text{if } a\doteq c
\end{array}
 \right. $\smallskip\\
 for $a \in \hat{\Sigma}, c \in \Sigma, p \in Q_2, r \in (Q_2 \cup \{-\})$
 \end{itemize}
 
The flush transition $\delta_{\text{3flush}}: Q_3 \times Q_3 \rightarrow 2^{Q_3}$ is defined by:
\begin{itemize}
\item $\delta_{\text{3flush}}(q_1, p_1) = \delta'_{\text{1flush}}(q_1, p_1), \forall q_1, p_1 \in Q'_1$, i.e. it simulates $\mathcal{A'}_1$ on $Q'_1$
\item $\delta_{\text{3flush}}(\langle \#, p, - \rangle, q) = \langle  \#, p, - \rangle$, with $p \in Q_2, q \in Q'_1$
\item $\delta_{\text{3flush}}(\langle  a_1, p_1, r_1 = p_2 \rangle, \langle a_2, p_2, r_2 \rangle) = \{\langle  a_2, q, r_2  \rangle \mid q \in \delta_{\text{2flush}}(p_1, p_2)\}$,\smallskip\\ where $a_1 \in \Sigma, a_2 \in \hat{\Sigma}$
\item $\delta_{\text{3flush}}(\langle a, p, r \rangle, q) = \{\langle  \#, s, - \rangle  \mid s \in \delta_{\text{2flush}}(p, r) \}$, for $a \in \Sigma, p, r \in Q_2, q \in Q'_1$\smallskip\\ i.e. whenever the precedence relations induce a merging of the subtrees of the words of the concatenation, $\mathcal{A}_3$ restores the state $s$ at the bottom of the stack of $\mathcal{A}_2$ from which a run of $\mathcal{A}_2$ will continue.
\end{itemize}
\end{itemize}

 It is clear that the \bfa\ $\mathcal{A}_3$ recognizes $L_1 \cdot L_2$, thus the class of languages accepted by \bfa\ is closed under concatenation on the left with languages recognized by OPAs.$\hfill\qed$
\end{proof}

\subsection*{Closure under complementation}
\label{sec:BooleanBFAF}

\begin{theorem} %[Closure of $\mathcal{L}$(\bfa) under complementation]
\label{th:clComplementBFAF}
Let $M$ be a conflict-free precedence matrix on an alphabet $\Sigma$.
Denote by $L_M\subseteq \Sigma^\omega$ the $\omega$-language comprising all infinite words $x\in\Sigma^\omega$ compatible with $M$.\\
Let $L$ be an $\omega$-language on $\Sigma$ that can be recognized by a nondeterministic \bfa\ with precedence matrix $M$ and $s$ states. Then the complement of $L$ w.r.t $L_M$ is recognized by an \bfa\ with the same precedence matrix $M$ and $2^{O(s^2)}$ states.
\end{theorem}

\begin{proof}
The proof follows to some extent the structure of the corresponding proof for B\"{u}chi VPAs~\cite{jacm/AlurM09}, but it exhibits some relevant technical aspects which distinctly characterize it; in particular, we need to introduce an ad-hoc factorization of $\omega$-words due to the more complex management of the stack performed by \ofa s.

Let $\mathcal A = \langle \Sigma, M, Q, I, F, \delta \rangle $ be a nondeterministic \bfa\ with $|Q| =s$. 
Without loss of generality $\mathcal A$ can be considered complete with respect to the transition function $\delta$, i.e. there is a run of $\mathcal{A}$ on every $\omega$-word on $\Sigma$ 
compatible with $M$.

 In general, % if one considers the structure of the words that can be parsed by an operator precedence $\omega$-automaton, one can notice that
 a sentence on $\Sigma^\omega$ compatible with $M$ can be factored in a unique way so as to distinguish the subfactors of the string that can be recognized without resorting to the stack of the automaton and those subwords for which the use of the stack is necessary.\\
 %This factorization is based on the definitions of simple chains, composed chains and pending letters introduced in Section~\ref{sec:newFormFA}.
% In order to define the factorization, it is necessary to introduce some notation, where simple chains, composed chains and their supports are defined as in~\cite{LonatiMandrioliPradella2011a}.
More precisely, an $\omega$-word $w \in \Sigma^\omega$ compatible with $M$ can be factored as a sequence of chains and pending letters $w = w_1 w_2 w_3 \ldots$ where either $w_i = a_i \in \Sigma$ is a pending letter or $w_i = a_{i1} a_{i2}\ldots a_{in}$ is a finite sequence of letters such that $\chain{l_i}{w_i}{first_{i+1}}$ is a chain,
% $l_{i} \lessdot a_{i1} a_{i2}\ldots a_{in} \gtrdot first_{i+1}$, and
where $l_i$ denotes the last pending letter preceding $w_i$ in the word and $first_{i+1}$ denotes the first letter of word $w_{i+1}$. Let also, by convention, $a_0 = \#$ be the first pending letter.

Notice that such factorization is not unique, since a string $w_i$ can be nested into a larger chain
having the same preceding pending letter. The factorization is unique, however, if we additionally require that $w_i$ has no prefix which is a chain.

As an example, for the word $w = \underbrace{\lessdot a \lessdot c \  \gtrdot} b \underbrace{\lessdot a \gtrdot} \underbrace{d \gtrdot} b \ldots$, with precedence relations in the OPM $a\gtrdot b$ and $b \lessdot d$, the unique factorization is $w = w_1 b w_3 w_4 b\ldots$, where $b$ is a pending letter and $\chain{\#}{a c}{b}, \chain{b}{a}{d}, \chain{b}{d}{b}$ are chains.

\noindent Define a \emph{semisupport for the simple chain}
$\chain {a_0} {a_1 a_2 \dots a_n} {a_{n+1}}$
as any path in $\mathcal A$ of the form
\begin{equation}
\label{eq:pseudosimplechain}
%q'_0
q_0
\va{a_1}{q_1}
\va{}{}
\dots
\va{}q_{n-1}
\va{a_{n}}{q_n}
\flush{q_0} {q_{n+1}}
%\va{a_{n+1}}{}.
\end{equation}
%Notice that the label of the last (and only) flush is exactly $q_0$, i.e. the first state of the path; this flush is executed because of relation $a_n
%\gtrdot a_{n+1}$.
%such that $\va{a_0}{q_0} = q'_0$ or there exists a chain $x_0$ such that  $\va{a_0}{q_0}\ourpath{x_0}{q'_0}$.

\noindent A \emph{semisupport for the composed chain}, with no prefix that is a chain, 
$\chain {a_0} {a_1 x_1 a_2 \dots a_n x_n} {a_{n+1}}$
is any path in $\mathcal A$ of the form
\begin{equation}
\label{eq:pseudocompchain}
%q'_0
q_0
%\ourpath{x_0}{q'_0}
\va{a_1}{q_1}
\ourpath{x_1}{q'_1}
\va{a_2}{}
\dots
\va{a_n} {q_n}
\ourpath{x_n}{q'_n}
%\flush{q'_0}{q_{n+1}}
\flush{q_0}{q_{n+1}}
%\va{a_{n+1}}{}
\end{equation}
where, for every $i: 1\leq i \leq n$: 
\begin{itemize}
\item if $x_i \neq \varepsilon$, then $\va{a_i}{q_i} \ourpath{x_i}{q'_i} $ 
is a support for the chain $\chain {a_i} {x_i} {a_{i+1}}$, i.e.,
it can be decomposed as $\va{a_i}{q_i} \ourpath{x_i}{q''_i} \flush{q_i}{q'_i}$.

\item if $x_i = \varepsilon$, then $q'_i = q_i$.
\end{itemize}

%\noindent Notice that the label of the last flush is exactly $q'_0$.
%\noindent and for the semisupport of both simple and composed chain above, we require that $\va{a_0}{q_0} = q'_0$ or there exists a chain $x_0$ such that  $\va{a_0}{q_0}\ourpath{x_0}{q'_0}$.

% Unlike the definition of the support for a simple (Equation~\ref{eq:simplechain}) and a composed chain (Equation~\ref{eq:compchain}), a semisupport for a chain represents the path visited during the parsing of the body of the chain including the flush moves induced by the symbol $a_{n+1}$, but the initial state $q'_0$ of the semisupport may be the state $q_0$ actually reached reading symbol $a_0$ or any state with which $q_0$ might be replaced after reading a chain $x_0$ between symbols $a_0$ and $a_1$.

 Unlike the definition of the support for a simple (Equation~\ref{eq:simplechain}) and a composed chain (Equation~\ref{eq:compchain}), in a semisupport for a chain the initial state $q_0$ is not restricted to be the state reached after reading symbol $a_0$.

Let $x\in \Sigma^*$ be such that $\chain a x b$ is a chain for some $a,b$ and let
$T(x)$ be the set of all triples $(q,p,f) \in Q \times Q \times \{0,1\}$ 
such that there exists a semisupport $q \ourpath{x} p$ in $\mathcal A$, and 
$f = 1$ iff the semisupport contains a state in $F$.
Also let $\mathcal T$ be the set of all such $T(x)$, i.e., $\mathcal T$ contains set of triples identifying all semisupports for some chain,  
%be the set $2^{Q \times Q\times \{0,1\}}$, i.e., 
%$\mathcal T$ contains sets of triples $(q,p,f)$ with $q,p\in Q, f \in \{0,1\}$,
and set  $PR = \Sigma \cup \mathcal{T}$. 
%The \textit{pseudorun} for $w$ in $\mathcal{A}$ is
$\mathcal{A}$'s \textit{pseudorun} for the word $w$, uniquely factorized as $ w_1 w_2 w_3 \ldots$ as stated above, is the $\omega$-word $w' = y_1 y_2 y_3\ldots \in PR^\omega$ 
where $y_i = a_i$ if $w_i = a_i$, otherwise $y_i = T(w_i)$.

 For the example above, then, %$w' = T(w_1) b T(w_3) T(w_4) b\ldots$.
$w' = T(a c)\ b\ T(a )\ T(d)\ b\ldots$.
%Now consider the language of $\omega$-words $z$ over alphabet $PR$ such that there exists $w \in L(\mathcal A)$ with $z = w'$. 

We now define a nondeterministic B\"{u}chi finite-state automaton $\mathcal{A}_R$ over alphabet $PR$ whose language includes the pseudorun $w'$ of any word $w \in L(\mathcal A)$. %that recognizes all, but not only, the pseudoruns $w'$ of words $w \in L(\mathcal A)$.
%as alphabet and recognizes all accepting pseudoruns, where a pseudorun is accepting if there is a run of $\mathcal{A}$ that parses the $\Sigma$ letters as usual and on letters $S_i \in \Sigma \times \mathcal T$ can switch the state from the first to the second component of a triple $(q_1, q_2,f)$ in $S_i$, and meets final states infinitely often.
%
$\mathcal{A}_R$ has all states of $\mathcal{A}$ and transitions corresponding to $\mathcal{A}$'s push transitions but it is devoid of flush edges (indeed they cannot be taken by a regular automaton without a stack). In addition,
for every $S \in \mathcal T$  it is endowed with arcs labeled $S$ which link, for each triple $(q,p,f)$ in $S$, 
either the pair of states $q, p$ or $q, p'$ if $f=1$, where $p'$ is a new final state which summarizes the states in $F$ met along the semisupport $ q\ourpath{} p $ and which has the same outgoing edges as $p$.

Notice that, given a set $S \in \mathcal T$, the existence of an edge $S$ between the pairs of states $q, p$ in the triples in $S$ can be decided in an effective way. 
%Deve esistere una parola w_i che ha esattamente le triple in S_I come supporto.
%In order to check whether there exists a chain for a symbol $S_i$ it suffices to consider, for each pair $q_1, q_2$ in $S_i$, a FA on the graph of $\mathcal{A}$ which recognizes the chains $\chain{a_i}{x}{b}$ , for any $b \in \Sigma$, having semisupport $ q_1\ourpath{} q_2 $. It is also necessary to keep track of the presence of final states along the runs parsing $x$, so as to compute for each pair $q_1, q_2$ the component $f$ properly. One can then build the FA recognizing the intersection of the languages accepted by these FAs 
%and can decide whether its words might have as supports pairs of states outside $S_i$ or not. 
%and can decide whether its words do not have as semisupports pairs of states outside $S_i$ (and thus, there exists an edge labeled $\pair{a_i}{S_i}$ between the pairs of states $q_1, q_2$ in the triples in $S_i$) or there are such semisupports (therefore there is no edge $\pair{a_i}{S_i}$ between the pairs of states $q_1, q_2$ in the triples in $S_i$).

 The automaton $\mathcal{A}_R$ built so far is able to parse all pseudoruns and recognizes all pseudoruns of $\omega$-words recognized by $\mathcal A$. 
%This is consistent with the fact that each word in $L_M$ corresponds to a pseudorun and a word belongs to $L = L(\mathcal{A}) $ iff the pseudorun is accepting. 
However, since its moves are no longer determined by the OPM $M$, it can also accept input words along the edges of the graph of $\mathcal{A}$ which are not pseudorun since they do not correspond to a correct factorization on $PR$. This is irrelevant, however, since the aim of the proof is to devise an automaton recognizing the complement of $L(\mathcal{A})$, and all the words in $L_M \backslash L(\mathcal{A})$ are parsed along pseudoruns, which are not accepted by $\mathcal{A}_R$. If one gives as input words only pseudoruns (and not generic words on $PR$), then they will be accepted by $\mathcal{A}_R$ if the corresponding words on $\Sigma$ belong to $L(\mathcal{A})$, and they will be rejected if the corresponding words do not belong to $L(\mathcal{A})$.
Given the B\"{u}chi finite-state automaton $\mathcal{A}_R$ (which has $O(s)$ states), one can now construct a deterministic Streett automaton $\mathcal{B}_R$ that accepts the complement of $L(\mathcal{A}_R)$, on the alphabet $PR$. If $\mathcal{B}_R$ receives as input words on $PR$ only pseudoruns, then it will accept only words in $L_M \backslash L(\mathcal{A})$. The automaton $\mathcal{B}_R$ has $2^{O(\text{s log s})}$ states and $O(s)$ accepting constraints~\cite{bib:Thomas1990a}.

Consider then a nondeterministic transducer \bfa\ $\mathcal{B}$ that on reading $w$  generates online the aforementioned pseudorun $w'$, which will be given as input to $\mathcal{B}_R$. The transducer $\mathcal{B}$ nondeterministically guesses whether 
the next input symbol is a pending letter, the beginning of a chain appearing in the factorization of $w$, 
or a symbol within such a chain, and uses stack symbols $Z$, $\bot$, or elements in $\mathcal T$, respectively,
to distinguish these three cases.

%an element of the factorization is a pending letter or a chain, pushing on the stack symbols $Z$ or $\bot$ respectively, and checks if the guess was correct. If the hypothesis fails $\mathcal{B}$ reaches an error state, otherwise it reaches final states infinitely often, accepting the input word.

In order to produce $w'$, whenever the automaton reads a pending letter it outputs the letter itself, whereas when it ends to recognize a chain of the factorization, performing a flush move towards a state with $\bot$ as first component, it outputs the set of all the pairs of states which define a semisupport for the chain.
%The transducer is formally defined in~\cite{PanellaMasterThesis2011}.
Thus, the output $w'$ produced by $B$ is unique, despite the nondeterminism of the translator.

Formally, the transducer \bfa\ $\mathcal{B} = \langle \Sigma, M, Q_B, I_B, F_B, PR, \delta_B,  \eta_B\rangle$ is defined as follows:
\begin{itemize}
\item $Q_B = %2^{(\{Z, \bot\} \cup Q) \times \hat{\Sigma} \times Q \times \{0,1\}}$ where $\hat{\Sigma} = \Sigma \cup \{\#\}$. 
\hat{\Sigma} \times \left( \{Z, \bot \} \cup \mathcal T \right)$ where $\hat{\Sigma} = \Sigma \cup \{\#\}$. 
The first component of a state in $Q_B$ denotes the lookback symbol read to reach the state, the second component 
represents the guess whether the next symbol to be read is a pending letter ($Z$), the beginning of a chain ($\bot$),
or a letter within such a chain $w_i$ ($T \in \mathcal T$). 
In the third case, $T$ contains all information necessary to correctly simulate the moves of $\mathcal A$ 
during the parsing of the chain $w_i$ of $w$,
and compute the corresponding symbol $y_i$ of $w'$.
In particular, $T$ is a set comprising all triples $(r,q,\nu)$ where 
$r$ represents the state reached before the last mark move, $q$ represents the current state reached by $\mathcal A$, 
and $\nu$ is a bit that reminds whether, while reading the chain, a state in $F$ has been encountered
(as in the construction of a deterministic OPA on words of finite length~\cite{LonatiMandrioliPradella2011b}, 
it is necessary to keep track of the state from which the parsing of a chain started, to avoid erroneous merges of runs on flush moves).

%The set of states of the automaton comprises tuples in $(\{Z, \bot \} \cup Q) \times \hat{\Sigma} \times Q  \times \{0,1\}$, where the first component of the tuple represents the guess that the next symbol to be read is a pending letter ($Z$) or the beginning of a chain ($\bot$) or it represents a state symbol that allows to perform correctly the nested flush moves during the parsing of a chain (as in the construction of a deterministic OPA on words of finite length~\cite{LonatiMandrioliPradella2011b}, it is necessary to keep track of the state from which the parsing of a chain started, in order to avoid erroneous merges of runs on flush moves). The other components of a tuple denote, respectively, the lookback symbol read to reach the state, the current state reached with the corresponding move on the graph of $\mathcal{A}$ and a bit that reminds whether, while reading a chain, a state in $F$ has been encountered,

\item $I_B = \left\{  \langle  \#, \bot \rangle ,  \langle  \#, Z \rangle \right\}$.
\item $F_B = \left\{  \langle  a, \bot \rangle ,  \langle  a, Z \rangle \mid a \in \hat\Sigma \right\}$.

%\item $ O = \Sigma \cup (\Sigma \times \mathcal T)$.

\item The transition function and the output function are defined as the union of two disjoint pairs of functions.
Let $a \in \hat\Sigma$, $b,c \in \Sigma$, $T,S \in \mathcal T$. 
%as\\ $\langle \delta_B, \eta_B \rangle : Q_B \times (\Sigma \cup Q_B) \rightarrow \mathcal{P}_F(Q_B \times O^*) $, where $\mathcal{P}_F$ denotes the set of finite subsets of $ (Q_B \times O^*) $.
The push pair $\langle \delta_{\text{Bpush}}, \eta_{\text{Bpush}} \rangle : Q_B \times \Sigma \rightarrow \mathcal{P}_F(Q_B \times PR^*) $ is defined as follows, where the symbols after $\downarrow$ denotes the output of the move of the automaton.

\begin{itemize}

\item \emph{Push of a pending letter.} 
%\item If $J = \{  \langle  Z, a, q, 0 \rangle \mid \forall q \in Q\}, a \in \hat{\Sigma}, b \in \Sigma$ then:

\[
\langle \delta_{\text{Bpush}}, \eta_{\text{Bpush}} \rangle \left( \langle a, Z \rangle, b \right) =
\left\{ 
\langle b, \bot \rangle \downarrow b, \  
\langle b, Z \rangle \downarrow b 
\right\}
\]
% \item If $J = \{  \langle  Z, a, q, 0 \rangle \mid \forall q \in Q\}, a \in \hat{\Sigma}, b \in \Sigma$ then:
% \[
% \langle \delta_{\text{Bpush}}, \eta_{\text{Bpush}} \rangle (J, b) =\{ \{\langle  \bot, b , q, 0 \rangle \mid \forall q \in Q\} \downarrow b, \{  \langle  Z, b, q, 0 \rangle \mid \forall q \in Q\} \downarrow b\}
% \]
% i.e. $b$ is guessed to be a pending letter and it is produced as output.

\item \emph{Mark at the beginning of a chain of the factorization.} 
If  $a \lessdot b$ then:
\[
\langle \delta_{\text{Bpush}}, \eta_{\text{Bpush}} \rangle \left( \langle a, \bot \rangle, b \right) = 
\left\{ \langle b, T \rangle \downarrow \varepsilon \right\}
\]
\[
\text{where } T = \left\{ 
\langle q,p,\nu \rangle \mid q \in Q, p \in \delta_{\text{push}} (q, b), \nu = 1 \text{ iff } p \in F  
\right\}
\]

\item \emph{Push within a chain of the factorization.} 
\[
\langle \delta_{\text{Bpush}}, \eta_{\text{Bpush}} \rangle \left( \langle a, T \rangle, b \right) = 
\left\{\langle b, S \rangle \downarrow \varepsilon \right\}
\quad \text{where} 
\]
\[
\!\!\!\!\!\!\!\!
S = \left\{ 
\langle t,p,\nu \rangle \mid
\exists \langle r,q,\xi \rangle \in T \text{ s.t. } 
t  =   
\left[  
\begin{array}{ll}
q & \text{if } a \lessdot b \\
r & \text{if } a\doteq b
\end{array}
 \right. , \
\nu = \left[
\begin{array}{ll}
\xi & \text{if } p \notin  F\\
1 & \text{if } p \in F
\end{array}
\right. , \
p \in \delta_{\text{push}} (q, b) 
\right\}
\]

% \item If $J = \{  \langle  \bot, a, q, 0 \rangle \mid \forall q \in Q\}, a \in \hat{\Sigma}, b \in \Sigma : M_{ab} = \lessdot$ or 
% $J \subseteq \{  \langle  r, a, s, x \rangle \mid r, s\in Q, x \in \{0,1\} \}, a \in \Sigma, b \in \Sigma$, then the automaton is reading the letters of a chain and does not output any symbol:\medskip\\
% $\begin{array}{ll}
% \langle \delta_{\text{Bpush}}, \eta_{\text{Bpush}} \rangle (J, b) = & \bigcup_{\substack{\langle r_i ,a, q_i, x_i \rangle \in J}} \{\langle  t, b, p, f \rangle \mid p \in \delta_{\text{push}}(q_i, b) \ \wedge \bigskip\\
% & \wedge \ t = \left\{  
% \begin{array}{ll}
% q_i & \text{if } a \lessdot b \\
% r_i & \text{if } a\doteq b
% \end{array}
%  \right. 
%  \wedge f = \left\{  
%  \begin{array}{ll}
%  x_i & \text{if } p \notin F\\
%  1 & \text{if } p\in F\\
%  \end{array}
%  \right.  \} \downarrow \varepsilon.
% \end{array}
% $
\end{itemize}

The flush pair $\langle \delta_{\text{Bflush}}, \eta_{\text{Bflush}} \rangle : Q_B \times Q_B \rightarrow \mathcal{P}_F(Q_B \times PR^*) $ is defined as follows.

\begin{itemize}

\item \emph{Flush at the end of a chain of the factorization.}
\[
\langle \delta_{\text{Bflush}}, \eta_{\text{Bflush}} \rangle ( \langle b, T \rangle, \langle a, \bot \rangle) = 
\left\{ 
\langle a, \bot \rangle \downarrow  R , \  
\langle a, Z \rangle \downarrow  R 
\right\}
\quad \text{where}
\]
\[
R = \left\{
\langle r,p, \nu \rangle \mid
\exists \langle r,q, \xi \rangle \in T,
\text{ s.t. } p \in \delta_{\text{flush}} (q, r), 
\nu = 
\left[ 
\begin{array}{ll}
\xi & \text{if } p \not\in F \\
1 & \text{if } p \in F
\end{array}
 \right. 
\right\}
\]

\item \emph{Flush within a chain of the factorization.}
\[
\langle \delta_{\text{Bflush}}, \eta_{\text{Bflush}} \rangle ( \langle b, T \rangle, \langle c, S \rangle) = 
\{ \langle c, R \rangle \downarrow \varepsilon \} 
\quad \text{where}
\]
\[
R = \left\{
\langle t,p, \nu \rangle \mid
\exists \langle r,q, \xi \rangle \in T,
\exists \langle t,r, \zeta \rangle \in S
\text{ s.t. } p \in \delta_{\text{flush}} (q, r), 
\nu = 
\left[ 
\begin{array}{ll}
\xi & \text{if } p \not\in F \\
1 & \text{if } p \in F
\end{array}
 \right. 
\right\}
\]

\end{itemize}

\noindent  An error state is reached for any other case. 
In particular, no flush move is defined when the second state has $Z$ as second component, 
nor when the first state has $Z$ or $\bot$ as second component, as consistent
with the meaning of stack symbol $Z$ and $\bot$.
%Furthermore, by definition of the initial states and the transition function, all the tuples in a state share the same second component (i.e. the lookback symbol). 
\end{itemize}

In the end, the final automaton to be built, which recognizes the complement of $L = L(\mathcal{A})$ w.r.t $L_M$, is the \bfa \ representing the product of $\mathcal{B}_R $ (converted to a B\"{u}chi automaton), which has $2^{O(\text{s log s})}$ states, and $\mathcal{B}$, which has $|Q_B| = 2^{O(s^2)}$ states: while reading $w$, $\mathcal{B}$ outputs the pseudorun $w'$ of $w$ online, and the states of $\mathcal{B}_R$ are updated accordingly. The automaton accepts if both $\mathcal{B}$ and $\mathcal{B}_R$ reach infinitely often final states. Furthermore, it has $2^{O(s^2)}$ states.~$\qed$
\end{proof}
\vspace{-0.1cm}

\subsection{Closure properties of \lof{\dbfa} under intersection and union}
\label{sec:wDBFA}

The class of languages accepted by \dbfa s is closed under intersection and union.

\subsection*{Closure under intersection}
\label{sec:intersectionDBFAF}

\begin{theorem}\label{th:clIntersectionDBFAF}
Let $L_1$ and $L_2$ be $\omega$-languages that can be recognized by two \dbfa s defined over the same alphabet $\Sigma$, with compatible precedence matrices $M_1$ and $M_2$ and $s_1$ and $s_2$ states respectively. Then $L = L_1 \cap L_2$ is recognizable by a \dbfa \  with OPM $M = M_1 \cap M_2$ and $O(s_1 s_2)$ states.
\end{theorem}

\begin{proof}
The proof derives from the analogous proof of closure with respect to intersection of languages recognized by \bfa s described in~\cite{PanellaMasterThesis2011}. In fact the \bfa \  which accepts the intersection of two languages $L_1$ and $L_2$ recognized by two \bfa s $\mathcal{A}_1$ and  $\mathcal{A}_2$ with compatible OPMs described in that proof is deterministic if both the automata $\mathcal{A}_1$ and $\mathcal{A}_2$ are deterministic.$\hfill\qed$
\end{proof}

\subsection*{Closure under union}
\label{sec:unionDBFAF}

\begin{theorem}\label{th:clUnionDBFAF}
Let $L_1$ and $L_2$ be $\omega$-languages that can be recognized by two \dbfa s defined over the same alphabet $\Sigma$, with compatible precedence matrices $M_1$ and $M_2$ and $s_1$ and $s_2$ states respectively. Then $L = L_1 \cup L_2$ is recognizable by an \dbfa \  with OPM $M = M_1 \cup M_2$ and $O(s_1 s_2)$ states.
\end{theorem}

\begin{proof}
Let $\tilde {\mathcal  A_{1}} = \langle \Sigma, M_1, \tilde Q_1, \tilde q_{01},\tilde F_1, \tilde \delta_1 \rangle $ and $\tilde {\mathcal  A_{2}} = \langle \Sigma, M_2, \tilde Q_2, \tilde q_{02}, \tilde F_2, \tilde \delta_2 \rangle $ be \dbfa s accepting the languages $L(\tilde {\mathcal A_1}) = L_1$ and $L(\tilde {\mathcal A_2}) = L_2$ and with compatible precedence matrices $M_1$ and $M_2$. Suppose without loss of generality that $\tilde Q_1$ and $\tilde Q_2$ are disjoint. Let $|\tilde Q_1| = s_1$ and $|\tilde Q_2| = s_2$.

Since $M_1$ and $M_2$ are compatible, then $M = M_1 \cup M_2$ is conflict-free and the two \dbfa s may be normalized completing their precedence matrix to $M = M_1 \cup M_2$ (see e.g. the normalization described in~\cite{PanellaMasterThesis2011}). The normalization preserves the determinism of the automata and keeps their sets of states disjoint.

The automata may be, then, completed as regards their transition function, so that there is a run on their graph for every $\omega$-word in $L_M$~\cite{PanellaMasterThesis2011}. The completed automata $\mathcal A_{1} = \langle \Sigma, M = M_1 \cup M_2, Q_1, q_{01}, F_1, \delta_1 \rangle $ and $\mathcal A_{2} = \langle \Sigma, M = M_1 \cup M_2, Q_2, q_{02}, F_2, \delta_2 \rangle $ are still deterministic with disjoint state sets and recognize the same languages as $\tilde {\mathcal  A_{1}}$ and $\tilde {\mathcal  A_{2}}$, i.e. $L(\mathcal A_1) = L_1$ and $L( \mathcal A_2) = L_2$. Furthermore, $|Q_1| = O(s_1)$ and $|Q_2| = O(s_2)$.

 An \dbfa \  $\mathcal{A}_3$ which recognizes $L_1 \cup L_2$ may then be defined adopting the usual product construction for regular automata: $\mathcal A_{3} = \langle \Sigma, M = M_1 \cup M_2, Q_3, q_{03}, F_3, \delta_3 \rangle $ where:
\begin{itemize}
\item $ Q_3 = Q_1 \times Q_2$,
\item $q_{03} = (q_{01}, q_{02})$,
\item $F_3 = F_1 \times Q_2  \cup Q_1 \times F_2$
\item and the transition function $\delta_3 : Q_3 \times ( \Sigma \cup  Q_3) \rightarrow  Q_3$ is defined as follows.
The push transition ${\delta}_{\text{3push}}:  Q_3 \times \Sigma \rightarrow  Q_3$ is expressed as:
\smallskip\\
$
{ \delta}_{\text{3push}}  ((q_1, q_2), a) = (\delta_{\text{1push}}(q_1,a), \delta_{\text{2push}}(q_2,a)) 
$\smallskip\\
$\forall q_1 \in Q_1, q_2 \in Q_2, a \in \Sigma$.

The flush transition ${\delta}_{\text{3flush}}:  Q_3 \times  Q_3 \rightarrow  Q_3$ is defined as:
\smallskip\\
$
{ \delta}_{\text{3flush}}  ((q_1, q_2), (p_1, p_2)) = (\delta_{\text{1flush}}(q_1, p_1), \delta_{\text{2flush}}(q_2,p_2))
$
\smallskip\\
$\forall q_1, p_1 \in Q_1, q_2, p_2 \in Q_2$
\end{itemize} 

The \dbfa \  $\mathcal{A}_3$ simulates $\mathcal{A}_1$ and $\mathcal{A}_2$ respectively on the two components of the states, and accepts an $\omega$-word iff there is an accepting run on it for at least one of the two automata.

The definition of the transition function is sound because the automata $\mathcal{A}_1$ and $\mathcal{A}_2$ have the same precedence matrix, thus they perform the same type of move (mark/push/\\flush) while reading the input word; furthermore, they are both complete w.r.t their transition function and none of them may stop a computation while reading a string.$\hfill\qed$

%The theorem derives immediately noting that, if an infinite word $x$ is accepted by at least one of the two automata, then it has a successful run on it and the word is also recognized by $\mathcal{A}_3$ on a run which passes through the pairs of states of the runs of $\mathcal{A}_1$ and $\mathcal{A}_2$ on $x$, performing the same type of move of the two automata for each input symbol and reaching infinitely often final states in $F_3$.

%Conversely, if the input word $x$ is accepted by $\mathcal{A}_3$, then there is a run on it which visits infinitely often final states, in $ F_1 \times Q_2$  or $ Q_1 \times F_2$ or both. Then, at least one of the runs of $\mathcal{A}_1$ and $\mathcal{A}_2$, which are the projections of the run of $\mathcal{A}_3$ on the first or second component respectively of the states, visits infinitely often final states and is accepting for $x$. $\hfill\qed$
\end{proof}

% We collect in Table~\ref{tab:closureProp} the closure properties of various families of \ofl s we proved, comparing them with the properties enjoyed by VPAs on infinite-length words. Binary operations are considered between languages with compatible OPMs.
%
%\vspace{-0.3cm}
%\begin{table}
%\centering
%\begin{tabular}{|c|c|c|c|c|}
%\hline
%& $\mathcal{L}$(\dbfa) & $\mathcal{L}$(\dmfa) & $\mathcal{L}$(\bfa)$\equiv$$\mathcal{L}$(\mfa) & $\mathcal{L}$($\omega$BVPA)~\cite{jacm/AlurM09} \\
% \hline 
%Intersection &Yes & Yes & Yes & Yes\\
% \hline
%Union &Yes & Yes & Yes & Yes\\
%\hline 
%Complement & No & Yes &Yes  & Yes\\
%\hline 
%$L_1 \cdot L_2$ & No & No &Yes  & Yes\\
%\hline 
%\end{tabular}\medskip
%\caption{Closure properties of families of $\omega$-languages. ($L_1 \cdot L_2$ denotes the concatenation of a language of finite-length words $L_1$ and an $\omega$-language $L_2$).\label{tab:closureProp}}
%\end{table}
%
%\vspace{-1.2cm}

\subsection{Closure properties of \lof{\bfae}}
\label{sec:wBFAE}

The class of languages accepted by \bfae s is closed under intersection and union, but not under complementation and concatenation on the left with an \opl.

\subsection*{Closure under intersection}
\label{sec:intersectionBFAE}

\begin{theorem}\label{th:clIntersectionBFAE}
Let $L_1$ and $L_2$ be $\omega$-languages that can be recognized by two \bfae s defined over the same alphabet $\Sigma$, with compatible precedence matrices $M_1$ and $M_2$ and $s_1$ and $s_2$ states respectively. Then $L = L_1 \cap L_2$ is recognizable by an \bfae \  with OPM $M = M_1 \cap M_2$ and $O(s_1 s_2)$ states.
\end{theorem}

\begin{proof}
Let $\mathcal{A}_1 = \langle \Sigma, M_1, Q_1, I_1, F_1, \delta_1 \rangle $ and $\mathcal{A}_2 = \langle \Sigma, M_2, Q_2, I_2, F_2, \delta_2 \rangle $ be \bfae s recognizing $L_1$ and $L_2$ respectively.

We can define for each \bfae\ an equivalent automaton \bfae\ whose set of states is partitioned into tagged states that are visited with empty stack and untagged states that are those visited with nonempty stack. This simple construction is described in~\cite{PanellaMasterThesis2011} to prove that \lof{\bfae}$\subseteq$ \lof{\bfa}, defining for each \bfae \ $\mathcal{A}$ an equivalent \bfa\ $\mathcal{\tilde{A}}$, but the resulting automaton $\mathcal{\tilde{A}}$ is still equivalent to $\mathcal{A}$  if it is interpreted as an \bfae. In particular the final states of the so built automaton are the tagged counterpart of the final states of the original \bfae.

Let $\tilde {\mathcal  A_{1}}$ and $\tilde {\mathcal  A_{2}}$ be \bfae \ equivalent to $\mathcal{A}_1$ and $\mathcal{A}_2$, respectively, defined following this construction. An \bfae\ $\mathcal{A}$ which recognizes $L_1 \cap L_2$ can be defined from $\tilde {\mathcal  A_{1}}$ and $\tilde {\mathcal  A_{2}}$ by resorting to the traditional approach to prove closure of regular B\"{u}chi automata under intersection, also adopted to prove closure under intersection for \bfa s.
The transformation of $\mathcal{A}_1$ and $\mathcal{A}_2$ into $\tilde {\mathcal  A_{1}}$ and $\tilde {\mathcal  A_{2}}$ guarantees that a run of $\mathcal{A}$ on an $\omega$-word reaches infinitely often a final state with empty stack iff both $\mathcal{A}_1$ and $\mathcal{A}_2$ have a run for the word which traverses infinitely often a final state with empty stack.$\hfill\qed$

\end{proof}

\subsection*{Closure under union}
\label{sec:unionBFAE}

\begin{theorem}\label{th:clUnionBFAE}
Let $L_1$ and $L_2$ be $\omega$-languages that can be recognized by two \bfae s defined over the same alphabet $\Sigma$, with compatible precedence matrices $M_1$ and $M_2$ and $s_1$ and $s_2$ states respectively. Then $L = L_1 \cup L_2$ is recognizable by an \bfae \  with OPM $M = M_1 \cup M_2$ and $O(|\Sigma|^2(s_1 + s_2))$ states.
\end{theorem}

\begin{proof}
The proof is analogous to the proof of closure under union for \bfa s.
More precisely, let $\tilde {\mathcal  A_{1}} = \langle \Sigma, M_1, \tilde Q_1, \tilde I_1,\tilde F_1, \tilde \delta_1 \rangle $ and $\tilde {\mathcal  A_{2}} = \langle \Sigma, M_2, \tilde Q_2, \tilde I_2, \tilde F_2, \tilde \delta_2 \rangle $ be \bfae s accepting the languages $L(\tilde {\mathcal A_1}) = L_1$ and $L(\tilde {\mathcal A_2}) = L_2$ and assume, without loss of generality, that $\tilde Q_1$ and $\tilde Q_2$ are disjoint. Let $|\tilde{Q}_1| = s_1$ and  $|\tilde{Q}_2| = s_2$.

Since $M_1$ and $M_2$ are compatible, then $M = M_1 \cup M_2$ is conflict-free and the two \bfae s may be normalized completing their OPM to $M = M_1 \cup M_2$ (see e.g. the normalization described in~\cite{PanellaMasterThesis2011}), obtaining two \bfae s $\mathcal  A_{1} = \langle \Sigma, M, Q_1, I_1, F_1, \delta_1 \rangle $ and $\mathcal A_{2} = \langle \Sigma, M, Q_2, I_2, F_2, \delta_2 \rangle $
which still recognize languages $L_1$ and $L_2$ respectively. The normalization keeps their sets of states disjoint.

The $\omega$-language $L = L_1 \cup L_2$ is recognized by the \bfae\ $\mathcal A = \langle \Sigma, M, Q = Q_1 \cup Q_2, I = I_1 \cup I_2, F = F_1 \cup F_2, \delta \rangle $ whose transition function
$\delta :  Q \times ( \Sigma \cup  Q) \rightarrow  2^Q$ is defined so as its restriction to $Q_1$ and $Q_2$ equals respectively $\delta_1 :  Q_1 \times ( \Sigma \cup  Q_1) \rightarrow  2^{Q_1}$ and $\delta_2 :  Q_2 \times ( \Sigma \cup  Q_2) \rightarrow  2^{Q_2}$, i.e for all $p, q \in Q, a \in \Sigma$:
\[\begin{array}{ll}
{ \delta}_{\text{push}}  (q, a) = \left\{  
\begin{array}{ll}
 { \delta_{\text{1push}}(q,a)} & \text{if } q \in Q_1\\
 { \delta_{\text{2push}}(q,a)} & \text{if } q \in Q_2\\
\end{array}
\right.
\smallskip\\
{ \delta}_{\text{flush}}  (p,q) = \left\{  
\begin{array}{ll}
 { \delta_{\text{1flush}}(p,q)} & \text{if } p, q \in Q_1\\
 { \delta_{\text{2flush}}(p,q)} & \text{if } p, q \in Q_2\\
\end{array}
\right.\\
\end{array}
\].

Hence, there exists a successful run in $\mathcal{A}$ on a word $x \in \Sigma^\omega$ iff there exists a successful run of $\mathcal A_1$ on $x$ or a successful run of $\mathcal A_2$ on $x$.
$\hfill\qed$
\end{proof}

\subsection*{Complementation and concatenation}
\label{sec:complConcatBFAE}

\begin{theorem}\label{th:clComplementBFAE}
Let $L$ be an $\omega$-language accepted by an \bfae\ with OPM $M$ on alphabet $\Sigma$. There does not necessarily exist an \bfae\ recognizing the complement of $L$ w.r.t $L_M$.
\end{theorem}

\begin{proof}
Let $M$ be a conflict-free OPM on alphabet $\Sigma = \{a,b\}$ given by:

$$
\begin{array}{c|cc}
      & a  & b\\
\hline
a & \lessdot & \lessdot  \\
b & \lessdot & \gtrdot  \\
\#       & \lessdot  &\lessdot\\
\end{array}
$$

Language $L = \{b^{\omega}\} \subseteq \Sigma^{\omega}$ is recognized by the \bfae\ with precedence matrix $M$ whose graph is represented in Figure~\ref{fig:ComplementBFAE}.
\begin{figure}[h!]
\begin{center}
\begin{tikzpicture}[every edge/.style={draw,solid}, node distance=4cm, auto, 
                    every state/.style={draw=black!100,scale=0.5}, >=stealth]

\node[initial by arrow, initial text=,state, accepting] (S) {{\huge $q_0$}};

\path[->]
(S) edge [loop right] node {$b$} (S)
(S) edge [loop below, double] node {$q_0$} (S);
\end{tikzpicture}
\caption{\bfae\ recognizing language $L$ of Theorem~\ref{th:clComplementBFAE}.}\label{fig:ComplementBFAE}
\end{center}
\end{figure}
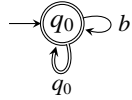
The complement of $L$ w.r.t $L_M$ includes words (with precedence relations between symbols defined by $M$) belonging to the set $\{a^n b^{\omega} \mid n \geq 1\}$ for which no \bfae\ can have an accepting run which reaches final states with empty stack infinitely often. 
$\hfill\qed$
\end{proof}

\begin{theorem}\label{th:clConcatBFAE}
Let $L_2$ be an $\omega$-language accepted by an \bfae\ with OPM $M$ on alphabet $\Sigma$ and let $L_1 \subseteq \Sigma^*$ be a language (of finite words) recognized by an \opa\ with a compatible precedence matrix. The $\omega$-language defined by the product $L_1 \cdot L_2$ is not necessarily recognizable by an \bfae.
\end{theorem}

\begin{proof}
Given $\Sigma = \{a,b\}$, let $L_1 = \{a^n \mid n \geq 1\}$ and let $L_2 = (L_{\text{Dyck}}(a,b))^{\omega}$ be the language of $\omega$-words composed by an infinite sequence of finite-length words belonging to the Dyck language with pair $a,b$.

$L_1$ is recognized by the \opa\ with OPM and graph in Figure~\ref{fig:Concat1BFAE} and language $L_2$ is recognized by the \bfae\ in Figure~\ref{fig:Concat2BFAE}.

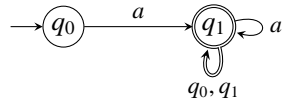
\begin{figure}[h!]
\begin{center}
\begin{tabular}{m{.4\textwidth}m{.45\textwidth}}
\qquad\quad
$
\begin{array}{c|cc}
      & a  & \# \\
\hline
a & \lessdot & \gtrdot  \\
\#       & \lessdot  &\doteq\\
\end{array}
$
&
\begin{tikzpicture}[every edge/.style={draw,solid}, node distance=4cm, auto, 
                    every state/.style={draw=black!100,scale=0.5}, >=stealth]

\node[initial by arrow, initial text=,state] (q0) {{\huge $q_0$}};
\node[state] (q1) [right of=q0, accepting, xshift=0cm] {{\huge $q_1$}};

\path[->]
(q0) edge [] node {$a$} (q1)
(q1) edge [loop right] node {$a$} (q1)
(q1) edge [loop below, double] node {$q_0, q_1$} (q1);
\end{tikzpicture}
\end{tabular}
\caption{\opa\ recognizing language $L_1$ of Theorem~\ref{th:clConcatBFAE}.}\label{fig:Concat1BFAE}
\end{center}
\end{figure}

\begin{figure}[h!]
\begin{center}
\begin{tabular}{m{.4\textwidth}m{.45\textwidth}}
\qquad\quad
$
\begin{array}{c|cc}
      & a & b  \\
\hline
a & \lessdot   & \doteq  \\
b  & \gtrdot & \gtrdot  \\
\#  & \lessdot  \\
\end{array}
$
&
\begin{tikzpicture}[every edge/.style={draw,solid}, node distance=4cm, auto, 
                    every state/.style={draw=black!100,scale=0.5}, >=stealth]

\node[initial by arrow, initial text=,state, accepting] (q0) {{\huge $q_0$}};
\node[state] (q1) [right of=q0, xshift=0cm] {{\huge $q_1$}};

\path[->]
(q0) edge [bend left] node {$a$} (q1)
(q1) edge [double] node {$q_0$} (q0)
(q1) edge [loop right] node {$a,b$} (q1)
(q1) edge [loop below, double] node {$q_1$} (q1);
\end{tikzpicture}
\end{tabular}
\caption{\bfae\ recognizing language $L_2$ of Theorem~\ref{th:clConcatBFAE}.}\label{fig:Concat2BFAE}
\end{center}
\end{figure}
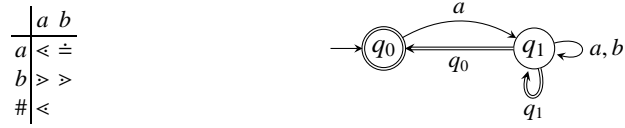

Language $L = L_1 \cdot L_2 = a^+ (L_{\text{Dyck}}(a,b))^{\omega}$, however, is not recognizable by any \bfae.
$\hfill\qed$
\end{proof}

\subsection{Closure properties of \lof{\dbfae}}
\label{sec:wDBFAE}

The class of languages accepted by \dbfae s is closed under intersection and union, but not under complementation and concatenation on the left with an \opl.

\subsection*{Closure under intersection}
\label{sec:intersectionDBFAE}

\begin{theorem}\label{th:clIntersectionDBFAE}
Let $L_1$ and $L_2$ be $\omega$-languages that can be recognized by two \dbfae s defined over the same alphabet $\Sigma$, with compatible precedence matrices $M_1$ and $M_2$ and $s_1$ and $s_2$ states respectively. Then $L = L_1 \cap L_2$ is recognizable by an \dbfae \  with OPM $M = M_1 \cap M_2$ and $O(s_1 s_2)$ states.
\end{theorem}

\begin{proof}
The proof derives from the analogous proof of closure under intersection of languages in \lof{\bfae} (Theorem~\ref{th:clIntersectionBFAE}). In fact, the transformation of \bfae s into equivalent \bfae s with tagged and untagged states preserves determinism and, similarly, the \bfae \ that accepts the intersection of the languages recognized by the two \bfae s $\tilde {\mathcal  A_{1}}$ and $\tilde {\mathcal  A_{2}}$ presented in that proof is deterministic if both $\tilde {\mathcal  A_{1}}$ and $\tilde {\mathcal  A_{2}}$ are deterministic.
$\hfill\qed$
\end{proof}

\subsection*{Closure under union}
\label{sec:unionDBFAE}

\begin{theorem}\label{th:clUnionDBFAE}
Let $L_1$ and $L_2$ be $\omega$-languages that can be recognized by two \dbfae s defined over the same alphabet $\Sigma$, with compatible precedence matrices $M_1$ and $M_2$ and $s_1$ and $s_2$ states respectively. Then $L = L_1 \cup L_2$ is recognizable by an \dbfae \  with OPM $M = M_1 \cup M_2$ and $O(s_1 s_2)$ states.
\end{theorem}

\begin{proof}
The proof is analogous to the proof of closure under union of languages belonging to \lof{\dbfa} (Theorem~\ref{th:clUnionDBFAF}).
$\hfill\qed$
\end{proof}

\subsection*{Complementation and concatenation}
\label{sec:complConcatDBFAE}

\begin{theorem}\label{th:clComplementDBFAE}
Let $L$ be an $\omega$-language accepted by an \dbfae\ with OPM $M$ on alphabet $\Sigma$. There does not necessarily exist an \dbfae\ recognizing the complement of $L$ w.r.t $L_M$.
\end{theorem}

\begin{proof}
Given $\Sigma = \{a,b\}$, the language $L = \{  \alpha \in \Sigma^\omega : \alpha $ contains an infinite number of letters $a\}$ can be recognized by an \dbfae\ $\mathcal  A =\langle \Sigma, M, Q, I, F, \delta \rangle $ with OPM and graph as in the figure below (Figure~\ref{fig:ComplementDBFAE}).
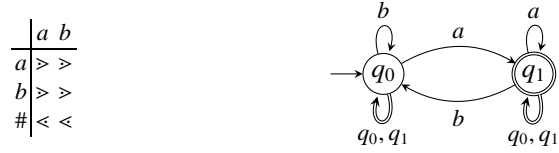
\begin{figure}[h!]
\begin{center}
\begin{tabular}{m{.4\textwidth}m{.45\textwidth}}
\qquad\quad
$
\begin{array}{c|cc}
      & a & b  \\
\hline
a & \gtrdot & \gtrdot   \\
b  & \gtrdot & \gtrdot  \\
\#       & \lessdot & \lessdot  \\
\end{array}
$
&
\begin{tikzpicture}[every edge/.style={draw,solid}, node distance=4cm, auto, 
                    every state/.style={draw=black!100,scale=0.5}, >=stealth]

\node[initial by arrow, initial text=,state] (q0) {{\huge $q_0$}};
\node[state] (q1) [right of=q0, accepting, xshift=0cm] {{\huge $q_1$}};

\path[->]
(q0) edge [bend left, above]  node {$a$} (q1)
(q0) edge [loop above] node {$b$} (q0)
(q0) edge [loop below, double] node {$q_0, q_1$} (q0)

(q1) edge [bend left, below]  node {$b$} (q0)
(q1) edge [loop above] node {$a$} (q1)
(q1) edge [loop below, double] node {$q_0, q_1$} (q1) ;
\end{tikzpicture}
\end{tabular}
\caption{ OPM and graph of the \dbfae\ $\mathcal A$ of Theorem~\ref{th:clComplementDBFAE}.}\label{fig:ComplementDBFAE}
\end{center}
\end{figure}

There is, however, no \dbfae\ that can recognize the complement of this language w.r.t. $L_M$, i.e. the language $\neg L = \{  \alpha \in \Sigma^\omega : \alpha \text{ contains finitely many letters $a$ } \}$.$\hfill\qed$
\end{proof}

\begin{theorem}\label{th:clConcatDBFAE}
Let $L_2$ be an $\omega$-language accepted by an \dbfae\ with OPM $M$ on alphabet $\Sigma$ and let $L_1 \subseteq \Sigma^*$ be a language (of finite words) recognized by an \opa\ with a compatible precedence matrix. The $\omega$-language defined by the product $L_1 \cdot L_2$ is not necessarily recognizable by an \dbfae.
\end{theorem}

\begin{proof}
Let $\Sigma = \{a,b\}$; the language $L$ of Equation~\ref{eq:Lfinitely_a} is the concatenation $L=L_1 \cdot L_2$ of a language of finite words $L_1$ and an $\omega$-language $L_2$, with compatible precedence matrices, defined as follows:

$L_1 = \Sigma^*$

$
L_2 \subseteq \Sigma^\omega,\quad L_2 = \{ b^\omega \}
$

\noindent Language $L_1$ is recognized by the \opa\ with OPM and state-graph in Figure~\ref{fig:Concat1DBFAE}:

\begin{figure}[h!]
\begin{center}
\begin{tabular}{m{.4\textwidth}m{.45\textwidth}}
\qquad\quad
$
\begin{array}{c|ccc}
      & a & b & \# \\
\hline
a & \lessdot & \lessdot  &\gtrdot \\
b  & \lessdot & \gtrdot &\gtrdot \\
\#       & \lessdot & \lessdot  &\doteq\\
\end{array}
$
&
\begin{tikzpicture}[every edge/.style={draw,solid}, node distance=4cm, auto, 
                    every state/.style={draw=black!100,scale=0.5}, >=stealth]

\node[initial by arrow, initial text=,state, accepting] (S) {{\huge $q_0$}};

\path[->]
(S) edge [loop above] node {$a, b$} (S)
(S) edge [loop right, double] node {$q_0$} (S);
\end{tikzpicture}
\end{tabular}
\caption{\opa\ recognizing language $L_1$ of Theorem~\ref{th:clConcatDBFAE}.}\label{fig:Concat1DBFAE}
\end{center}
\end{figure}
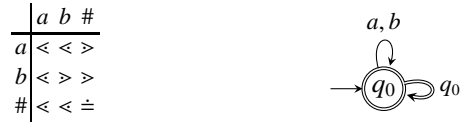

\noindent and language $L_2$ is recognized by the \dbfae\ in Figure~\ref{fig:Concat2DBFAE}:

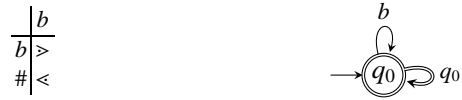
\begin{figure}[h!]
\begin{center}
\begin{tabular}{m{.4\textwidth}m{.45\textwidth}}
\qquad\quad
$
\begin{array}{c|c}
      & b  \\
\hline
b   & \gtrdot  \\
\#  & \lessdot  \\
\end{array}
$
&
\begin{tikzpicture}[every edge/.style={draw,solid}, node distance=4cm, auto, 
                    every state/.style={draw=black!100,scale=0.5}, >=stealth]

\node[initial by arrow, initial text=,state, accepting] (S) {{\huge $q_0$}};

\path[->]
(S) edge [loop above] node {$b$} (S)
(S) edge [loop right, double] node {$q_0$} (S);
\end{tikzpicture}
\end{tabular}
\caption{\dbfae\ recognizing language $L_2$ of Theorem~\ref{th:clConcatDBFAE}.}\label{fig:Concat2DBFAE}
\end{center}
\end{figure}
\noindent Since language $L$ cannot be recognized by an \dbfae, then the class of languages \lof{\dbfae} is not closed w.r.t concatenation.
$\hfill\qed$
\end{proof}

\section{Conclusions and further research}
\label{sec:conclusions}

We presented a formalism for infinite-state model checking based on operator precedence languages, continuing to explore the paths in the lode of operator precedence languages started up by Robert Floyd a long time ago. 
We introduced various classes of automata able to recognize operator precedence languages of infinite-length words whose expressive power outperforms classical models for infinite-state systems as Visibly Pushdown $\omega$-languages, allowing to represent more complex systems in several practical contexts. We proved the closure properties of \ofl s under Boolean operations that, along with the decidability of the emptiness problem, are fundamental for the application of such a formalism to model checking.
For instance, with reference to Example~\ref{ex:BFA}, imagine that one builds a specialized system that includes only procedures of type $a$ and where interrupts of lowest level are disabled when there is any pending $call_a$: once having built a new model $\hat{\mathcal{A}}$ for such a system she can automatically verify its compliance with the more general one $\mathcal{A}$ by checking whether $L(\hat{\mathcal{A}}) \subseteq L(\mathcal{A})$.

Our results open further directions of research. 
A first topic deals with the investigation of properties and fields of application of OPAs and \ofa s as transducers, as they may e.g. translate tagged documents written in mark-up languages (as XML, HTML) into the final displayed (XML, HTML) page, or they may translate the traces of operations of do-undo actions performed on different versions of a file into an end-user log or document. Thus, it might be possible to define a formal translation from structured or semistructured languages or patterns of tasks and client behaviors into suitable final-user views of the model.

A second interesting research issue is the characterization of \ofl s in terms of suitable monadic second order logical formulas, that has already been studied for operator precedence languages of finite-length strings~\cite{LonatiMandrioliPradella2013a}.
This would further strengthen applicability of model checking techniques.
The next step of investigation will regard the actual design and study
of complexity issues of algorithms for model checking of expressive
logics on these pushdown models. We expect that the peculiar features
of operator precedence languages, as their ``locality principle'' which makes them
suitable for parallel and incremental parsing~\cite{BarenghiEtAl2012a,BVCMP12} and their expressivity, might be interestingly exploited to devise efficient and attractive software model-checking procedures and approaches.

\bibliographystyle{splncs03}
\bibliography{Floydbib}

\end{document}